\documentclass[12pt]{article}
\usepackage{amsthm}
\usepackage{eurosym}
\usepackage{amsfonts}
\usepackage{amssymb,amsmath}
\usepackage[dvips]{color}
\usepackage{amscd}
\usepackage{tikz}
\usepackage{mathtools}
\usepackage{ytableau}
\usepackage{youngtab}
\usepackage{subcaption}
\usepackage{pgfopts}

\usetikzlibrary{decorations.pathreplacing}

\usetikzlibrary{patterns}

\def\and{\mathrm{and}}

\newtheorem{lemma}{Lemma}
\newtheorem{prop}{Proposition}

\newtheorem{thm}{Theorem}

\newcommand{\be}{\begin{equation}}
\newcommand{\ee}{\end{equation}}
\newcommand{\bea}{\begin{eqnarray}}
\newcommand{\eea}{\end{eqnarray}}
\newcommand{\beas}{\begin{eqnarray*}}
\newcommand{\eeas}{\end{eqnarray*}}
\newcommand{\ba}{\begin{array}}
\newcommand{\ea}{\end{array}}

\newcommand{\nbox}{{\,\lower0.9pt\vbox{\hrule \hbox{\vrule height 0.2 cm \hskip 0.19 cm \vrule height 0.2 cm}\hrule}\,}}

\def\href#1#2{#2}
\begin{document}

\begin{titlepage}
\hfill
\vbox{
    \halign{#\hfil         \cr
           } 
      }  

\hbox to \hsize{{}\hss \vtop{ \hbox{}

}}

\vspace*{20mm}
\begin{center}

{\large \textbf{A construction of quarter BPS coherent states and Brauer}}

{\normalsize \vspace{5mm} }

{\large \textbf{algebras} }

{\normalsize \vspace{5mm} }

{\large \textbf{} }

{\Large \vspace{ 20mm} }

{\normalsize {Hai Lin${}^{1}$, Keyou Zeng${}^{2}$}  }

{\normalsize \vspace{10mm} }

{\small \emph{${}^1$\textit{Yau Mathematical Sciences Center, Tsinghua University,
Beijing 100084, P. R. China
}} }

{\normalsize \vspace{0.2cm} }

{\small \emph{$^2$\textit{Department of Physics, Tsinghua University,
Beijing 100084, P. R. China
\\
}} }

{\normalsize \vspace{0.4cm} }

\end{center}

\begin{abstract}

BPS coherent states closely resemble semiclassical states and they have gravity dual descriptions in terms of semiclassical geometries. The half BPS coherent states have been well studied, however less is known about quarter BPS coherent states. Here we provide a construction of quarter BPS coherent states. They are coherent states built with two matrix fields, generalizing the half BPS case. These states are both the eigenstates of annihilation operators and in the kernel of dilatation operator. Another useful labeling of quarter BPS states is by representations of Brauer algebras and their projection onto a subalgebra $\mathbb{C}[S_n\times S_m]$. Here, the Schur-Weyl duality for the Walled Brauer algebra plays an important role in organizing the operators. One interesting subclass of these Brauer states are labeled by representations involving two Young tableaux. We obtain the overlap between quarter BPS Brauer states and quarter BPS coherent states, where the Schur polynomials are used. We also derive superposition formulas transforming a truncated version of quarter BPS coherent states to quarter BPS Brauer states. The entanglement entropy of Brauer states as well as the overlap between Brauer states and squeezed states are also computed.

\end{abstract}

\end{titlepage}

\vskip 1cm

\section{Introduction}

\label{sec_Introduction}\renewcommand{\theequation}{1.\arabic{equation}} %
\setcounter{equation}{0}

The gauge/gravity correspondence \cite%
{Maldacena:1997re,Gubser:1998bc,Witten:1998qj} has provided a remarkable way
to study quantum gravity by quantum field theory on the boundary of the
spacetime. It nontrivially relates a quantum system without gravity to a
quantum theory with gravity. This correspondence reveals the emergence of
spacetime geometry from the degrees of freedom on the boundary. The bulk
spacetime dynamically emerges from the boundary quantum mechanical system
\cite{Rangamani:2016dms,VanRaamsdonk:2010pw,Horowitz:2006ct,Koch:2009gq}.
This duality further provides us a way to explore interesting quantitative
features of non-perturbative effects in string theory and quantum gravity,
since it allows us to perform calculations pertaining to the gravity side by
working in the quantum field theory side.

\vspace{1pt}

Based on very general arguments involving the supersymmetry algebra, one can
calculate exactly the correlation functions of some special class of
operators in $\mathcal{N}=4$ SYM, namely the BPS operators. The properties
of being protected from quantum correction make these quantities valuable in
the study of the field theory and their relation with the dual gravity side.
People now have a rather clear understanding of the dynamics of the half BPS
operators and the dual gravity picture. In the gravity dual, there are
backreacted geometries that correspond to these BPS states in the field
theory side \cite{Corley:2001zk,Berenstein:2004kk,Lin:2004nb}.

\vspace{1pt}

Coherent state \cite{Zhang:1990fy,Quantum information} is a very important
concept in quantum mechanics, often describing a state that most closely
resembles the behavior of a semiclassical state. Coherent states arise in a
wide range of physical systems and have applications in different fields
such as quantum optics. It was realized previously \cite{Berenstein:2005aa}
that coherent states also play an important role in the context of
gauge/gravity duality. The half BPS coherent states have been constructed,
and their various properties have been studied. And in \cite%
{Berenstein:2017abm}, the half BPS coherent states in $\mathcal{N}=4$ SYM
are related with the phenomenon of topology change in quantum gravity.

\vspace{1pt}

However, the analog of the half BPS coherent states for the quarter BPS
generalization has not been completely known in the literature. In this
paper, we will extend the previous definition of half BPS coherent states to
give a construction of quarter BPS coherent states, and study its relation
with other operators in the quarter BPS sector, such as the Brauer operators
\cite{Kimura:2007wy,Kimura:2010tx}. Then it is natural to proceed to
consider more complicated quarter BPS states. However, a systematic
understanding of quarter BPS operators is not an easy task. In this paper,
we will consider the large $N$ limit where the analysis gets simplified
since the dilation operator \cite{Beisert:2003tq,Kimura:2010tx} can be
simplified. And the construction of quarter BPS operators provides us
necessary ingredients to construct quarter BPS coherent states, which is a
main topic in this paper.

\vspace{1pt}

The quarter BPS states \cite{DHoker:2003csh,Ryzhov:2001bp} play important
roles in gauge/gravity duality. Apart from the multi-trace basis, the
quarter BPS states have representation bases, including the Brauer basis,
the restricted Schur basis, and the flavor symmetry basis. They have been
reviewed in for example \cite{Ramgoolam:2008yr,Koch:2009gq}. In addition to
these labelings, the quarter BPS coherent states serve as another labeling
of the operators or states in the Hilbert space. These states live in the
same Hilbert space, hence one can superpose them and compute transition
probabilities between states, and such operations have been performed in
\cite{Berenstein:2017abm,Diaz:2015tda,Brown:2006zk,Lin:2017dnz}. Different
states can be distinguished from each other, by carefully observing
correlation functions \cite%
{Skenderis:2007yb,Christodoulou:2016nej,Skenderis:2006di,Balasubramanian:2007qv}%
.

\vspace{1pt}

Schur-Weyl duality and its generalizations provide us a powerful set of
tools to organize the gauge theory operators and to relate them with
interesting configurations in the dual gravity theory. In many previous
examples, the representation of the symmetric group is used to construct
gauge invariant operators that have interesting geometric interpretation
\cite{Lin:2004nb,Corley:2001zk,Berenstein:2017abm,Lin:2017dnz}. In this
paper, we make use of the generalization of Schur-Weyl duality involving
Walled Brauer algebras, see \cite{Ramgoolam:2008yr,Kimura:2007wy}. The
Brauer algebras and Walled Brauer algebras \cite%
{Brauer:1937,Benkart:1994,Halverson:1996} can be regarded as a
generalization of the group algebra of symmetric group and play a similar
role in constructing gauge invariant operators. We will call the operators
labeled by representations of Brauer algebra, Brauer operators, and the
corresponding states in the Hilbert space, Brauer states. A subclass of
Brauer operators gives us useful examples of quarter BPS operators. Brauer
states share many features that are similar to Young tableau states in \cite%
{Corley:2001zk,Berenstein:2017abm,Lin:2017dnz}.

\vspace{1pt}

We will consider three sets of labelings of the states: the trace product
basis, the coherent states, and the Brauer states. Their relations will be
studied. We will calculate the overlap of these different states. It is
interesting that we obtain many results involving the Brauer states which
are very similar with our previous results of Young tableau states in \cite%
{Lin:2017dnz}. Besides, we can write superposition formulas transforming
between quarter BPS Brauer states and quarter BPS coherent states that
resemble our previous superposition formulas that describe topology change
\cite{Berenstein:2017abm,Lin:2017dnz}.

\vspace{1pt}

The organization of this paper is as follows. In Section 2, we describe
quarter BPS operators of the $\mathcal{N}=4$ SYM in the large $N$ limit, and
based on these results, we construct the quarter BPS coherent states. Then
in Section 3, we analyze the Brauer states and their relation to other
states including quarter BPS coherent states. Afterwards in Section 4, we
analyze the squeezed states that generalize the coherent states. In Section
5, we discuss our results and draw conclusions. Finally, we include
Appendices A and B for more details on Brauer algebras.

\vspace{1pt}

\section{Construction of quarter BPS coherent states}

\renewcommand{\theequation}{2.\arabic{equation}} \setcounter{equation}{0} %
\renewcommand{\thethm}{2.\arabic{thm}} \setcounter{thm}{0} %
\renewcommand{\theprop}{2.\arabic{prop}} \setcounter{prop}{0}

\label{sec_Young tableau states and entanglement}

\vspace{1pt}

\subsection{Trace product basis}

To begin with, we describe the general form of a multi-trace operator built
from two complex scalar fields. We can make use of symmetric group, since
the trace structure can be captured by a permutation $\alpha \in S_{n+m}$.
We consider operators of the form
\begin{equation}
\text{tr}(\alpha Z^{\otimes n}\otimes Y^{\otimes m})=Z_{i_{\alpha
(1)}}^{i_{1}}Z_{i_{\alpha (2)}}^{i_{2}}\cdots Z_{i_{\alpha
(n)}}^{i_{n}}Y_{i_{\alpha (n+1)}}^{i_{n+1}}Y_{i_{\alpha
(n+m)}}^{i_{n+m}},\quad \alpha \in S_{n+m}.
\end{equation}%
For example, for the case $n=m=2$, $\alpha =1$ corresponds to $(\text{tr}%
Z)^{2}(\text{tr}Y)^{2}$, $\alpha =(1234)$ corresponds to $\text{tr}%
(Z^{2}Y^{2})$, and $\alpha =(1324)$ corresponds to $\text{tr}(ZYZY)$. Note
that if two permutations $\alpha ,\alpha ^{\prime }$ are conjugate to each
other by an element $h\in S_{n}\times S_{m}$, they correspond to the same
state $\text{tr}(\alpha Z^{\otimes n}\otimes Y^{\otimes m})=\text{tr}(\alpha
^{\prime }Z^{\otimes n}\otimes Y^{\otimes m})$. Therefore we make use of $%
S_{n}\times S_{m}$ equivalence class of $S_{n+m}$: Two elements $\alpha
,\alpha ^{\prime }\in S_{n+m}$ are in the same $S_{n}\times S_{m}$
equivalence class $[\alpha ]=[\alpha ^{\prime }]$ if and only if $\alpha
=h\alpha h^{-1}$ for some $h\in S_{n}\times S_{m}$. As a special case, we
note that for $m=0$, this basis gives us all multi-trace operators, and is a
basis of the half BPS operators, see \cite{Berenstein:2017abm}. More
precisely, the basis is labeled by conjugacy class of $S_{n}$, which is just
given by a sequence $(w_{1},w_{2},\dots )$ where $w_{s}$ means that there
are $w_{s}$ cycles of length $s$ in the conjugacy class. The Brauer
operators, which we will introduce later, can be expanded by the above
basis. See Appendix A of \cite{Kimura:2007wy} for some examples. Therefore
in the following, we will consider all states labeled by $S_{n}\times S_{m}$
equivalence class of $S_{n+m}$.

\vspace{1pt}

The Hilbert space of all two-matrix multi-trace operators has a tensor
product structure. For the case of half BPS operators, the Hilbert space has
a tensor product structure given by the momentum number $k$ and $\mathcal{H}%
=\bigotimes_{k}\mathcal{H}_{k}$, where each $\mathcal{H}_{k}$ is created by $%
a_{k}^{\dagger }\leftrightarrow \text{tr}(\frac{Z}{\sqrt{N}})^{k}$
corresponding to a single trace operator. Write the operator by a
permutation, then a single trace corresponds to a permutation that has only
one cycle. For the more general two-matrix case, we also expect that a
factor of the tensor product is created by a single trace operator. A
general single trace operator can be written as
\begin{equation}
\text{tr}(\prod_{j}Z^{n_{j}}Y^{m_{j}}).
\end{equation}%
Note that since $\text{tr}(Z^{k_{1}}Y^{k_{2}}\cdots Z^{k_{n}})=\text{tr}%
(Z^{k_{1}+k_{n}}Y^{k_{2}}\cdots )$, we can always take a state into the
standard form $\text{tr}(\prod_{j}Z^{n_{j}}Y^{m_{j}})$ where $%
n_{j},m_{j}\geq 1$. To label these states, we define the following sets
\begin{eqnarray}
K_{0} &=&\{\vec{k}=(k_{1},0)\text{ or }(0,k_{2})\mid k_{1},k_{2}\geq 1\}, \\
K_{2} &=&\{\vec{k}=(k_{1},k_{2})\mid k_{1},k_{2}\geq 1\}, \\
K_{4} &=&\{\vec{k}=(k_{1},k_{2},k_{3},k_{4})\mid k_{1},k_{2},k_{3},k_{4}\geq
1\}, \\
&&\cdots  \notag \\
K_{2p} &=&\{\vec{k}=(k_{1},\dots ,k_{2p})\mid k_{i}\geq 1,\;i=1,\dots
,2p\},~\cdots
\end{eqnarray}%
However, there is some redundancy in the above labeling since the traces
have cyclic invariant property. For example, consider $\vec{k}=(1,1,3,1)$
and $\vec{k}^{\prime }=(3,1,1,1)$, they are different as vectors in $K_{4}$.
However
\begin{equation}
\text{tr}(ZYZ^{3}Y)=\text{tr}(Z^{3}YZY),
\end{equation}%
by cyclic property of trace. More generally, if two $\vec{k}$ are related by
a cyclic permutation, they actually define the same multi-trace operator.
Consider a group action of the abelian group $\mathbb{Z}_{p}$ on $K_{2p}$
for $p\geq 1$. Write $\lambda $ for the generator of $\mathbb{Z}_{p}$, or in
other words $\mathbb{Z}_{p}=\langle \lambda \rangle $ with $\lambda ^{p}=id$%
. Then define the action of $\lambda $ on $K_{2p}$

\begin{equation}
\lambda \cdot (k_{1},k_{2},\dots
,k_{2p-1},k_{2p})=(k_{2p-1},k_{2p},k_{1},k_{2},\dots ,k_{2p-3},k_{2p-2}).
\end{equation}%
And we define the quotient $\tilde{K}_{2p}=K_{2p}/\mathbb{Z}_{p}~$where two $%
\vec{k},\vec{k}^{\prime }$ are equivalent to each other if and only if $\vec{%
k}=\lambda ^{l}\cdot \vec{k}^{\prime }$ for some $l$. Then we define the
equivalence class $[\vec{k}]$, where from the above $[\vec{k}^{\prime }]=[%
\vec{k}]$. Note that $\mathbb{Z}_{1}=\{id\}$ is trivial, therefore $K_{2}=%
\tilde{K_{2}}$, and also $K_{0}=\tilde{K_{0}}$.$~$Then we define the set
\begin{equation}
\tilde{K}=K_{0}\cup K_{2}\cup \tilde{K}_{4}\cup \tilde{K}_{6}\cup \cdots
\end{equation}%
Then the Hilbert space can be written as
\begin{equation}
\mathcal{H}=\bigotimes_{[\vec{k}]\in \tilde{K}}\mathcal{H}_{[\vec{k}]},
\end{equation}%
where $\mathcal{H}_{[\vec{k}]}$ is spanned by $\left[ \text{tr}\left( \frac{Z%
}{\sqrt{N}}\right) ^{k_{1}}\left( \frac{Y}{\sqrt{N}}\right) ^{k_{2}}\cdots %
\right] ^{w_{[\vec{k}]}},w_{[\vec{k}]}=0,1,2,\dots $. Now we identify the
operator:
\begin{equation}
a_{[\vec{k}]}^{\dagger }\;\leftrightarrow ~O_{[\vec{k}]}=\;\text{tr}\left[
\left( \frac{Z}{\sqrt{N}}\right) ^{k_{1}}\left( \frac{Y}{\sqrt{N}}\right)
^{k_{2}}\cdots \right] .  \label{identify_basis}
\end{equation}%
A complete trace product basis is then provided by%
\begin{equation}
\prod_{\lbrack \vec{k}]\in \tilde{K}}a_{[\vec{k}]}^{\dagger w_{[\vec{k}%
]}}\rvert 0\rangle .  \label{basis_01}
\end{equation}

Using this notation, the half BPS case corresponds to $\bigotimes_{[\vec{k}%
]\in K_{0}}\mathcal{H}_{[\vec{k}]}$. In the half BPS case, states $%
\prod_{k}a_{k}^{\dagger w_{k}}\rvert 0\rangle $ with $\sum_{k}kw_{k}=n$
correspond to multi-trace operators labeled by a equivalence class of $S_{n}$%
. In the two matrix case, states $\prod_{[\vec{k}]\in \tilde{K}}a_{[\vec{k}%
]}^{\dagger w_{[\vec{k}]}}\rvert 0\rangle $ with $\sum_{[\vec{k}%
]}((\sum_{i}k_{2i+1})w_{[\vec{k}]})=n,\;\sum_{[\vec{k}]}((\sum_{i}k_{2i})w_{[%
\vec{k}]})=m$ correspond to multi-trace operators labeled by a $S_{n}\times
S_{m}$ equivalence class of $S_{n+m}$. That is, a sequence%
\begin{equation}
\mathbf{w}=(w_{[\vec{k}]})_{[\vec{k}]\in \tilde{K}}
\end{equation}%
with $\sum_{[\vec{k}]}((\sum_{i}k_{2i+1})w_{[\vec{k}]})=n,\;\sum_{[\vec{k}%
]}((\sum_{i}k_{2i})w_{[\vec{k}]})=m$ uniquely corresponds to a $S_{n}\times
S_{m}$ equivalence class of $S_{n+m}$. And each $w_{[\vec{k}]}$ is the
number of cycles of type $[\vec{k}]$ in the permutation. We use notation $%
\mathbf{w}$ to distinguish our previous notation $\vec{w}=(w_{1},w_{2},\dots
)$ in the half BPS case, see for example \cite%
{Berenstein:2017abm,Lin:2017dnz}.

The inner product is a very important structure for the Hilbert space. Let $%
\alpha ,\beta \in S_{n+m}$, with $n,m\geq 1$. we consider the inner product
\begin{eqnarray}
\langle \text{tr}(\alpha Z^{\otimes n}\otimes Y^{\otimes m})^{\dagger }\text{%
tr}(\beta Z^{\otimes n}\otimes Y^{\otimes m})\rangle &=&\sum_{\sigma \in
S_{n}\times S_{m}}\prod_{k=1}^{n}\delta _{i_{\beta (\sigma
(k))}}^{i_{k}}\delta _{i_{\alpha ^{-1}(k)}}^{i_{\sigma
(k)}}\prod_{h=n+1}^{n+m}\delta _{i_{\beta (\sigma (h))}}^{i_{h}}\delta
_{i_{\alpha ^{-1}(h)}}^{i_{\sigma (h)}}  \notag \\
&=&\sum_{\sigma \in S_{n}\times S_{m}}\text{tr}(\alpha ^{-1}\sigma
^{-1}\beta \sigma )  \notag \\
&=&\sum_{\sigma \in S_{n}\times S_{m}}N^{C(\alpha ^{-1}\sigma ^{-1}\beta
\sigma )},
\end{eqnarray}%
where in the second line, we sum over $\sigma \in S_{n}\times S_{m}$. And in
the last line we use $C(\alpha )$ to represent the number of cycles in the
permutation $\alpha $. Note that this expression contains all orders of $N$,
and highest order of $N$ appears only when $\alpha ^{-1}\sigma ^{-1}\beta
\sigma =1$, $C(1)=n+m$. Therefore $N^{m+n}$ appear only when $\alpha ,\beta $
are in the same equivalence class, that is $[\alpha ]=[\beta ]$. Using
appropriate normalization:
\begin{eqnarray}
\langle \text{tr}(\alpha Z^{\otimes n}\otimes Y^{\otimes m})^{\dagger }\text{%
tr}(\beta Z^{\otimes n}\otimes Y^{\otimes m})\rangle &=&\sum_{\sigma \in
S_{n}\times S_{m}}N^{C(\alpha ^{-1}\sigma ^{-1}\beta \sigma )}  \notag \\
&=&N^{m+n}(\delta _{\lbrack \alpha ],[\beta ]}(\sum_{\substack{ \sigma \in
S_{n}\times S_{m}  \\ \alpha =\sigma ^{-1}\alpha \sigma }}1)+O(1/N)).  \notag
\\
&&
\end{eqnarray}%
And in the last line, $\delta _{\lbrack \alpha ],[\beta ]}$ equals $1$ only
when $\alpha ,\beta $ are in the same $S_{n}\times S_{m}$ equivalence class
of $S_{n+m}$, which is also when $\text{tr}(\alpha Z^{\otimes n}Y^{\otimes
m})=\text{tr}(\beta Z^{\otimes n}Y^{\otimes m})$.

The coefficient $\sum_{\substack{ \sigma \in S_{n}\times S_{m}  \\ \alpha
=\sigma ^{-1}\alpha \sigma }}1$ is determined by the $S_{n}\times S_{m}$
equivalence class of $\alpha $. The equivalence class $[\alpha ]$ can be
represented by the sequence $\mathbf{w}$, then after some calculation, we
have
\begin{equation}
\sum_{\substack{ \sigma \in S_{n}\times S_{m}  \\ \alpha =\sigma ^{-1}\alpha
\sigma }}1=\prod_{[\vec{k}]\in K}N([\vec{k}],w_{[\vec{k}]}),
\end{equation}%
where
\begin{equation}
N([\vec{k}],w_{[\vec{k}]})=%
\begin{cases}
k^{w_{k}}w_{k}!~~~~~[\vec{k}]\in K_{0},\vec{k}=(k,0)\text{ or }\vec{k}=(0,k)
\\
w_{[\vec{k}]}!~~~~~[\vec{k}]\notin K_{0}%
\end{cases}%
.
\end{equation}%
Therefore we have
\begin{equation}
\langle \text{tr}(\alpha (\frac{Z}{\sqrt{N}})^{\otimes n}\otimes (\frac{Y}{%
\sqrt{N}})^{\otimes m})^{\dagger }\text{tr}(\beta (\frac{Z}{N})^{\otimes
n}\otimes (\frac{Y}{\sqrt{N}})^{\otimes m})\rangle =\delta _{\lbrack \alpha
],[\beta ]}\prod_{[\vec{k}]\in \tilde{K}}N([\vec{k}],w_{[\vec{k}]})+O(1/N),
\label{matrix_inner}
\end{equation}%
where we have worked in the large $N$ limit. The above formula completely
determines the inner product in the Hilbert space of quarter BPS operators.
Using our notation (\ref{basis_01}), the inner product can be written as
\begin{equation}
\langle 0\lvert a_{[\vec{k}]}^{w_{[\vec{k}]}}a_{[\vec{k}]}^{\dagger w_{[\vec{%
k}]}}\rvert 0\rangle =N([\vec{k}],w_{[\vec{k}]})+O(1/N),
\end{equation}%
and we get the commutation relation for the operators
\begin{equation}
\lbrack a_{[\vec{k}]},a_{[\vec{k}]^{\prime }}^{\dagger }]=N([\vec{k}%
],1)\delta _{\lbrack \vec{k}],[\vec{k}^{\prime }]}+O(1/N).
\end{equation}

The trace product basis is very useful for the computation of inner products
of other operators, since we can expand other operators in terms of trace
product basis.

\subsection{Coherent states in a general form}

In the rest of this Section, we work in the infinite $N$ limit. With the
notation of Sec. 2.1, a general coherent state $\rvert Coh\rangle ~$is
\begin{equation}
\exp (\sum_{[\vec{k}]\in \tilde{K}}c_{[\vec{k}]}a_{[\vec{k}]}^{\dagger
})\rvert 0\rangle .
\end{equation}%
The condition for coherent state is
\begin{equation}
a_{[\vec{k}]}\rvert Coh\rangle =c_{[\vec{k}]}\rvert Coh\rangle \;\text{ for }%
[\vec{k}]\in \tilde{K},
\end{equation}%
and we need to check this. Using the commutation relations%
\begin{equation}
\lbrack a_{[\vec{k}]},a_{[\vec{k}^{\prime }]}^{\dagger }]=%
\begin{cases}
k\delta _{\lbrack \vec{k}],[\vec{k}^{\prime }]}~~[\vec{k}]\in K_{0},[\vec{k}%
]=(k,0)\text{ or }[\vec{k}]=(0,k) \\
\delta _{\lbrack \vec{k}],[\vec{k}^{\prime }]}~~[\vec{k}]\notin K_{0}%
\end{cases}%
,
\end{equation}%
\begin{equation}
\lbrack a_{[\vec{k}]},a_{[\vec{k}^{\prime }]}]=0,
\end{equation}%
we have that
\begin{eqnarray}
a_{[\vec{k}]}\rvert Coh\rangle &=&a_{[\vec{k}]}\prod_{[\vec{l}]\in \tilde{K}%
}\exp (c_{[\vec{l}]}a_{[\vec{l}]}^{\dagger })\rvert 0\rangle  \notag \\
&=&\prod_{\substack{ \lbrack \vec{l}]\in \tilde{K}  \\ \lbrack \vec{l}]\neq
\lbrack \vec{k}]}}\exp (c_{[\vec{l}]}a_{[\vec{l}]}^{\dagger })\left( a_{[%
\vec{k}]}\sum_{j}\frac{(c_{[\vec{k}]}a_{[\vec{k}]}^{\dagger })^{j}}{j!}%
\right) \rvert 0\rangle  \notag \\
&=&\prod_{\substack{ \lbrack \vec{l}]\in \tilde{K}  \\ \lbrack \vec{l}]\neq
\lbrack \vec{k}]}}\exp (c_{[\vec{l}]}a_{[\vec{l}]}^{\dagger })\left( N([\vec{%
k}],1)c_{[\vec{k}]}\sum_{j}\frac{j(c_{[\vec{k}]}a_{[\vec{k}]}^{\dagger
})^{j-1}}{{j!}}\right) \rvert 0\rangle  \notag \\
&=&N([\vec{k}],1)c_{[\vec{k}]}\rvert Coh\rangle ,
\end{eqnarray}%
where in the third line we have used $[a_{[\vec{k}]},a_{[\vec{k}]}^{\dagger
j}]=N([\vec{k}],1)ja_{[\vec{k}]}^{\dagger j-1}$.~And thus we have proved
that the state $\exp (\sum_{[\vec{k}]\in \tilde{K}}c_{[\vec{k}]}a_{[\vec{k}%
]}^{\dagger })\rvert 0\rangle $ is a coherent state.

The coherent states have overlaps with the trace product states. Write a
general trace product state as
\begin{equation}
\prod_{\lbrack \vec{k}]\in \tilde{K}}a_{[\vec{k}]}^{\dagger w_{[\vec{k}%
]}}\rvert 0\rangle .
\end{equation}%
We consider overlap
\begin{eqnarray}
\langle 0\lvert \prod_{\lbrack \vec{k}]\in \tilde{K}}a_{[\vec{k}]}^{w_{[\vec{%
k}]}}\rvert Coh\rangle &=&\langle 0\lvert \prod_{\lbrack \vec{k}]\in \tilde{K%
}}a_{[\vec{k}]}^{w_{[\vec{k}]}}\prod_{[\vec{k}]\in K}\exp (c_{[\vec{k}]}a_{[%
\vec{k}]}^{\dagger })\rvert 0\rangle  \notag \\
&=&\prod_{[\vec{k}]\in \tilde{K}}\sum_{j}\frac{(c_{[\vec{k}]})^{j}}{j!}%
\langle 0\lvert a_{[\vec{k}]}^{w_{[\vec{k}]}}a_{[\vec{k}]}^{\dagger j}\rvert
0\rangle  \notag \\
&=&\prod_{[\vec{k}]\in \tilde{K}}\frac{(c_{[\vec{k}]})^{w_{[\vec{k}]}}}{w_{[%
\vec{k}]}!}N([\vec{k}],w_{[\vec{k}]}),
\end{eqnarray}%
where
\begin{equation}
N([\vec{k}],w_{[\vec{k}]})=%
\begin{cases}
k^{w_{k}}w_{k}!~~~[\vec{k}]\in K_{0},[\vec{k}]=(k,0)\text{ or }[\vec{k}%
]=(0,k) \\
w_{[\vec{k}]}!~~~[\vec{k}]\notin K_{0}%
\end{cases}%
.
\end{equation}%
We will use notation $k\leftrightarrow (k,0),\bar{k}\leftrightarrow (0,k)$,
and define
\begin{equation}
c_{k}=\frac{\Lambda _{k}}{k},\quad c_{\bar{k}}=\frac{\Lambda _{\bar{k}}}{%
\bar{k}}.
\end{equation}%
We use this definition to make our results to be compatible with our
previous half BPS case \cite{Lin:2017dnz}. Then the above results can be
written as
\begin{equation}
\langle 0\lvert \prod_{\lbrack \vec{k}]\in \tilde{K}}a_{[\vec{k}]}^{w_{[\vec{%
k}]}}\rvert Coh\rangle =\prod_{k=1}^{\infty }\Lambda _{k}^{w_{k}}\prod_{\bar{%
k}=1}^{\infty }\Lambda _{\bar{k}}^{w_{\bar{k}}}\prod_{[\vec{k}]\in \tilde{K}%
-K_{0}}(c_{[\vec{k}]})^{w_{[\vec{k}]}}.  \label{coh_trace_product_02}
\end{equation}%
\vspace{1pt}\vspace{1pt}

We then consider inner product between coherent states. To simplify
notation, we define
\begin{equation}
\rvert Coh_{1}\rangle =\exp (\sum_{[\vec{k}]\in \tilde{K}}\alpha _{\lbrack
\vec{k}]}a_{[\vec{k}]}^{\dagger })\rvert 0\rangle ,\quad \rvert
Coh_{2}\rangle =\exp (\sum_{\vec{k}\in \tilde{K}}\beta _{\lbrack \vec{k}%
]}a_{[\vec{k}]}^{\dagger })\rvert 0\rangle .
\end{equation}%
And we also define
\begin{equation}
\alpha _{k}=\frac{\Lambda _{k}}{k},\quad \alpha _{\bar{k}}=\frac{\Lambda _{%
\bar{k}}}{\bar{k}},\quad \beta _{k}=\frac{B_{k}}{k},\quad \beta _{\bar{k}}=%
\frac{B_{\bar{k}}}{\bar{k}}.
\end{equation}%
Then the overlap between two coherent states is
\begin{eqnarray}
\langle Coh_{2}\lvert Coh_{1}\rangle &=&\sum_{\{w_{[\vec{k}]}\}}\langle
0\lvert \prod_{\lbrack \vec{k}]\in \tilde{K}}\frac{(\beta _{\lbrack \vec{k}%
]}^{\ast })^{w_{[\vec{k}]}}}{w_{[\vec{k}]}!}a_{[\vec{k}]}^{w_{[\vec{k}%
]}}\rvert Coh_{1}\rangle  \notag \\
&=&\sum_{\{w_{[\vec{k}]}\}}\prod_{[\vec{k}]\in \tilde{K}}\frac{(\beta
_{\lbrack \vec{k}]}^{\ast })^{w_{[\vec{k}]}}}{w_{[\vec{k}]}!}\frac{(\alpha
_{\lbrack \vec{k}]})^{w_{[\vec{k}]}}}{w_{[\vec{k}]}!}N([\vec{k}],w_{[\vec{k}%
]})  \notag \\
&=&\sum_{\{w_{[\vec{k}]}\}}\prod_{k}\frac{(B_{k}^{\ast }\Lambda _{k})^{w_{k}}%
}{k^{w_{k}}w_{k}!}\prod_{\bar{k}}\frac{(B_{\bar{k}}^{\ast }\Lambda _{\bar{k}%
})^{w_{\bar{k}}}}{\bar{k}^{w_{\bar{k}}}w_{\bar{k}}!}\prod_{[\vec{k}]\in
\tilde{K}-K_{0}}\frac{(\beta _{\lbrack \vec{k}]}^{\ast }\alpha _{\lbrack
\vec{k}]})^{w_{[\vec{k}]}}}{w_{[\vec{k}]}!}  \notag \\
&=&\exp (\sum_{k}\frac{1}{k}B_{k}^{\ast }\Lambda _{k}+\sum_{\bar{k}}\frac{1}{%
\bar{k}}B_{\bar{k}}^{\ast }\Lambda _{\bar{k}}+\sum_{[\vec{k}]\in \tilde{K}%
-K_{0}}\beta _{\lbrack \vec{k}]}^{\ast }\alpha _{\lbrack \vec{k}]}).
\label{Inner_general_Coh}
\end{eqnarray}%
We can consider the case when
\begin{equation}
\Lambda _{k}=x_{1}^{k}+x_{2}^{k}+\dots ,\;B_{k}=y_{1}^{k}+y_{2}^{k}+\dots
,\;\Lambda _{\bar{k}}=x_{\bar{1}}^{\bar{k}}+x_{\bar{2}}^{\bar{k}}+\dots
,\;B_{\bar{k}}=y_{\bar{1}}^{\bar{k}}+y_{\bar{2}}^{\bar{k}}+\dots .
\end{equation}%
Then
\begin{eqnarray}
\exp (\sum_{k}\frac{1}{k}B_{k}^{\ast }\Lambda _{k}) &=&\exp (\sum_{k}\frac{1%
}{k}(y_{1}^{\ast k}+y_{2}^{\ast k}+\dots )(x_{1}^{k}+x_{2}^{k}+\dots ))
\notag \\
&=&\prod_{i,j}\exp (\sum_{k}\frac{(x_{i}y_{j}^{\ast })^{k}}{k})  \notag \\
&=&\prod_{i,j}\frac{1}{1-x_{i}y_{j}^{\ast }}.
\end{eqnarray}%
\qquad This result agrees perfectly with our previous results about half BPS
operators \cite{Lin:2017dnz}.

\vspace{1pt}

\subsection{Quarter BPS coherent states and dilatation operator}

In \cite{DHoker:2003csh,Ryzhov:2001bp}, a systematic construction of quarter
BPS operators is given, and the construction is based on the following
symmetrized trace operator
\begin{equation}
A_{r_{1}\cdots r_{p}}=\text{tr}(W_{(r_{1}}\cdots W_{r_{p})})
\end{equation}%
where $r_{i}=1,2$ and we assume that $W_{1}=Z,W_{2}=Y$. Since all the
indices are symmetrized, and the operator $A_{r_{1}\cdots r_{p}}$ only
depends on how many $r_{i}$ are $1$ and how many $r_{i}$ are $2$, we write
\begin{equation}
A_{(n,m)}=A_{\underbrace{1,\dots ,1}_{n},\underbrace{2,\dots ,2}_{m}}
\end{equation}%
with $n+m=p$.

By definition
\begin{equation}
A_{r_{1}\cdots r_{p}}=\sum_{\sigma \in S_{p}}\frac{1}{p!}\text{tr}%
(W_{r_{\sigma (1)}}\cdots W_{r_{\sigma (p)}}).
\end{equation}%
Therefore $A_{(n,m)}$ is a linear combination of our previously defined
operators in Section 2.1:
\begin{equation}
A_{(n,m)}=\sum_{[\vec{k}]\in \tilde{K}^{(n,m)}}C([\vec{k}])\text{tr}%
(Z^{k_{1}}Y^{k_{2}}Z^{k_{3}}Y^{k_{4}}\dots )
\end{equation}%
where $C([\vec{k}])$ is some constant to be determined, and
\begin{equation}
\tilde{K}^{(n,m)}=\{[\vec{k}]\in \tilde{K}\mid
\sum_{i}k_{2i-1}=n,\;\sum_{i}k_{2i}=m\}.
\end{equation}

However, the labeling $\tilde{K}$ is inconvenient for our later study, since
we will consider the action of symmetric group on the operators, which takes
a complicated form using the labeling $\tilde{K}$. Therefore, we will
introduce a new but equivalent labeling of single trace operators. We see
that the sequence of matrices $W_{r}=Z,Y$ resemble a $1d$ lattice, where
each lattice site has spin up and down. So we define the set
\begin{equation}
V_{n,m}=\{f\in \mathrm{Map}(\{1,\dots ,n+m\},\{\uparrow ,\downarrow \})\mid
n=\#\{i\mid f(i)=\uparrow \},\;m=\#\{i\mid f(i)=\downarrow \}\}.
\end{equation}%
There is a left action of the group $S_{n+m}$ on $V_{n,m}$:
\begin{equation}
(\sigma f)(i)=f(\sigma ^{-1}(i)),\;\forall \sigma \in S_{n+m}.
\end{equation}%
Here we use $\sigma ^{-1}$ to define this action to make the action
compatible with the group structure. We can check that
\begin{equation}
((\sigma _{1}\sigma _{2})f)(i)=f(\sigma _{2}^{-1}\sigma
_{1}^{-1}(i))=(\sigma _{1}(\sigma _{2}f))(i).
\end{equation}%
The cyclic group $\mathbb{Z}_{n+m}$ is a subgroup of $S_{n+m}$, therefore $%
\mathbb{Z}_{n+m}\subset S_{n+m}$ acts on $V_{n,m}$. Denote $\lambda $ the
generator of $\mathbb{Z}_{n+m}$, which satisfies $\lambda ^{n+m}=1$. We can
write the group action as
\begin{equation}
(\lambda ^{k}f)(i)=f(i-k).
\end{equation}%
The cyclic property of trace tells us that lattice configurations that
differ only by a cyclic permutation should be considered as equivalent.
Hence we define the quotient $\tilde{V}_{n,m}=V_{n,m}/\mathbb{Z}_{n+m}$,$~$%
where we identify $f\sim \lambda ^{k}f$, and write $[f]$ for the
corresponding equivalence class in $\tilde{V}_{n,m}$. For each $[f]\in
\tilde{V}_{n,m}$, we have a corresponding operator%
\begin{equation}
O_{[f]}=\text{tr}(\prod_{i=1}^{n+m}W_{f_{i}})  \label{O_f}
\end{equation}%
where we identify $W_{\uparrow }=Z,W_{\downarrow }=Y$. By the cyclic
property of the trace, the definition of $O_{[f]}$ does not depend on the
choice of representative $f$. The advantage of using $\tilde{V}$ to label
the single trace operator is that we can easily write down the action of
symmetric group $S_{n+m}$ on the operators. First, the action of $S_{n+m}$
on $V_{n,m}$ induces an action on the quotient $\tilde{V}_{n,m}$, and
further induces an action on the operator labeled by $\tilde{V}_{n,m}$:
\begin{equation}
\sigma O_{[f]}=O_{[\sigma f]},\quad \forall \sigma \in S_{n+m}.
\end{equation}

Having shown the advantage of using $\tilde{V}$, we still need to show why
this labeling is equivalent to our previous one $\tilde{K}$. To see this, we
establish a bijection between the two sets: $\varphi :\tilde{K}^{(n,m)}%
\xrightarrow{\sim}\tilde{V}_{n,m}$. On the one hand, for a $[\vec{k}]\in
\tilde{K}^{(n,m)}$, and write $\vec{k}=(k_{1},k_{2},\dots ,k_{2j})$, we
define the corresponding $\varphi ([\vec{k}])=[f]=[\uparrow
^{k_{1}}\downarrow ^{k_{2}}\uparrow ^{k_{3}}\downarrow ^{k_{4}}\cdots ]$.
The $f=\uparrow ^{k_{1}}\downarrow ^{k_{2}}\uparrow ^{k_{3}}\downarrow
^{k_{4}}\cdots $ is defined by the following:
\begin{equation}
f(1)=\cdots f(k_{1})=\uparrow ,\;f(k_{1}+1)=\cdots f(k_{1}+k_{2})=\downarrow
,\;f(k_{1}+k_{2}+1)=\dots f(k_{1}+k_{2}+k_{3})=\uparrow ,\dots .
\label{V_K_relation}
\end{equation}%
On the other hand, for every $[f]\in \tilde{V}_{(n,m)}$, we can define the
corresponding $[\vec{k}]=\varphi ^{-1}([f])\in \tilde{K}^{(n,m)}$ by the
following:
\begin{equation}
\varphi ^{-1}([f])\in \tilde{K}^{(n,m)}=%
\begin{cases}
\lbrack (n,0)],~~~~{\text{ if }}[f]=[\uparrow ^{n}]\in \tilde{V}_{n,0} \\
\lbrack (0,m)],~~~{\text{ if }}[f]=[\downarrow ^{m}]\in \tilde{V}_{0,m} \\
\lbrack (k_{1},k_{2},k_{3},k_{4},\dots )],~~~{\text{ if }}[f]=[\uparrow
^{k_{1}}\downarrow ^{k_{2}}\uparrow ^{k_{3}}\downarrow ^{k_{4}}\cdots ]\in
\bigcup_{n,m\geq 1}\tilde{V}_{n,m}%
\end{cases}%
.
\end{equation}

The above formula is complete since $[\uparrow ^{n}]$ is the only element in
$\tilde{V}_{n,0}$, and $[\downarrow ^{m}]$ is the only element in $\tilde{V}%
_{0,m}$, while every element in $\bigcup_{n,m\geq 1}\tilde{V}_{n,m}$ can be
written in the form $[\uparrow ^{k_{1}}\downarrow ^{k_{2}}\cdots ]$.

Then we can move on to consider the dilatation operator, it will turn out
that our above notation will give a simple expression for the operator.
According to \cite{Beisert:2003tq,Kimura:2010tx}, the dilatation operator
can be expanded in power series of the coupling constant as
\begin{equation}
\hat{D}=\sum_{l=0}(\frac{g_{YM}^{2}}{16\pi ^{2}})^{l}\hat{D}_{2l},
\end{equation}%
where $\hat{D}_{2l}$ is the $l$ loop dilatation operator. We also consider
the $N\rightarrow \infty $ limit. Since the coupling constants are related
by $g_{YM}^{2}N=\lambda $, where $\lambda $ is the~'t Hooft coupling, we can
write
\begin{equation}
\hat{D}=\hat{D}_{0}+\frac{\lambda }{16\pi ^{2}N}\hat{D}_{2}.
\end{equation}%
Therefore, for one loop dilatation operator acting on some operator, we only
consider results in the $N\rightarrow \infty $ limit, or in other words, we
will consider:
\begin{equation}
\frac{\lambda }{16\pi ^{2}N}\hat{D}_{2}O(Z,Y).
\end{equation}

The zero loop operator can be written as $\hat{D}_{0}=\text{tr}(Z\frac{d}{dZ}%
+Y\frac{d}{dY})$. And one can easily calculate that the action of the zero
loop dilation operator on a general multi-trace operator $\text{tr}(\alpha
Z^{\otimes n}\otimes Y^{\otimes m})$ for $\alpha \in S_{n+m}$ is:
\begin{equation}
\hat{D}_{0}\text{tr}(\alpha Z^{\otimes n}\otimes Y^{\otimes m})=(n+m)\text{tr%
}(\alpha Z^{\otimes n}\otimes Y^{\otimes m}).
\end{equation}%
We can write the above formula in our previous notation
\begin{equation}
\hat{D}_{0}\prod_{[\vec{k}]\in \tilde{K}}a_{[\vec{k}]}^{\dagger w_{[\vec{k}%
]}}\rvert 0\rangle =\left( \sum_{[\vec{k}]\in \tilde{K}}w_{[\vec{k}%
]}(\sum_{i}k_{i})\right) \prod_{[\vec{k}]\in \tilde{K}}c_{[\vec{k}%
]}^{\dagger w_{[\vec{k}]}}\rvert 0\rangle ,
\end{equation}%
or we can write
\begin{equation}
\lbrack \hat{D}_{0},a_{[\vec{k}]}^{\dagger }]=\sum_{i}k_{i}.
\end{equation}

Now for the one loop operator:
\begin{equation}
\hat{D}_{2}=-2\text{tr}([Z,Y][\frac{d}{dZ},\frac{d}{dY}]),
\end{equation}%
our aim is to calculate
\begin{equation}
\frac{\lambda }{16\pi ^{2}N}\hat{D}_{2}O_{[f]}=\frac{\lambda }{16\pi ^{2}N}%
\hat{D}_{2}\text{tr}(\prod_{i=1}^{n+m}W_{f_{i}}).
\end{equation}%
In the $N\rightarrow \infty $ limit, we write $\hat{D}_{2}^{\prime }=\frac{1%
}{N}\hat{D}_{2}$, which can be shown to satisfy the Leibniz rule. To write
down the $\hat{D}_{2}^{\prime }$ we first need to introduce some special
permutation operators that will be useful, namely the swap operator $%
P_{i,i+1}\in S_{n+m}$, which is defined by:
\begin{equation}
P_{i,i+1}(i)=i+1,\;P_{i,i+1}(i+1)=i,\;P_{i,i+1}(j)=j,\text{ for }j\neq i,i+1,
\label{swap_operator}
\end{equation}%
where we identify the indices $n+m+1\sim 1.$ Then using the notations (\ref%
{O_f})-(\ref{swap_operator}), the one loop dilatation operator can be easily
written as
\begin{equation}
\frac{1}{2}\hat{D}_{2}^{\prime }=\sum_{i=1}^{n+m}(\mathbb{I}-P_{i,i+1}).
\label{D_2_01}
\end{equation}

Another advantage of using $\tilde{V}$ for the labeling is that the
symmetrized trace operator $A_{(n,m)}$ can be written in a simple form:
\begin{equation}
A_{(n,m)}=\frac{1}{(n+m)!}\sum_{\sigma \in S_{n+m}}\sigma O_{[f_{0}]},
\end{equation}%
where $[f_{0}]\in \tilde{V}_{n,m}$ can be given by $f_{0}(1)=\cdots
=f_{0}(n)=\uparrow ,f_{0}(n+1)=\cdots =f_{0}(n+m)=\downarrow $.

Here as before, we use $a^{\dagger }$ to represent the corresponding
creation operator, for example
\begin{equation}
a_{[f]}^{\dagger }\text{ }\leftrightarrow O_{[f]},\quad a_{(n,m)}^{\dagger
}\leftrightarrow A_{(n,m)}.
\end{equation}%
These are analogous to the relation (\ref{identify_basis}).

\vspace{1pt}

\begin{prop}
The following operator defined by
\begin{equation}
O_{(n,m)}=\sum_{\substack{ l_{i}\geq 0  \\ l_{1}+l_{2}+\dots +l_{m}=n}}\text{%
tr}(Z^{l_{1}}YZ^{l_{2}}Y\cdots Z^{l_{m}}Y)
\end{equation}%
for $m>0$, and $O_{(n,0)}=\text{tr}(Z^{n})$ for $m=0$, is proportional to
the symmetrized operator $A_{(n,m)}$.
\end{prop}

\begin{proof}
	For $m > 0 $, we can write $O_{(n,m)}$ as
	\begin{equation}
	O_{(n,m)} = \sum_{[f] \in \tilde{V}_{n,m}}c_{[f]}O_{[f]}
			\end{equation}
	and we need to find the coefficient $c_{[f]}$.
	
	First, we see that for any $[f] \in \tilde{V}_{(n,m)}$, $c_{[f]} \geq 1$. This is because we can describe a $[f] \in \tilde{V}_{(n.,m)}$ by saying that there are first $l_1$ $Z$s followed by a $Y$ and then $l_2$ $Z$s followed by a $Y$ and so on.
	
	Second, for a generic sequence $l_1,l_2,\dots,l_{m}$, the new sequence $l_1' = l_{\lambda(1)},l_2' = l_{\lambda(2)}, \dots,$ $l_{m}' = l_{\lambda(m)}$ for a $\lambda \in \mathbb{Z}_{m}$ corresponds to the same operator. For each $\lambda \in \mathbb{Z}_{m}$ we have a new but equivalent sequence, thus if we forget the over counting, a generic operator $O_{[f]}$ is counted $m$ times in the summation. However, there is possibility of over counting since it is possible that for some $l_1,l_2,\dots,l_{m}$ and some $\lambda$ we have $l_1 = l_{\lambda(1)},l_2 = l_{\lambda(2)}, \dots, l_{m} = l_{\lambda(m)}$. Therefore for general cases, the results should be $m$ divided by the number of group elements in $\mathbb{Z}_{m}$ that keeps the sequence $l_1,\dots,l_m$ unchanged. Remember that we use $[f] \in \tilde{V}_{(n,m)}$ to label the operator, and if the sequence $l_1,\dots,l_m$ corresponds to $[f]$, then the number of group elements in $\mathbb{Z}_{m}$ that fixes the sequence $l_1,\dots,l_m$ is the same as $|Stab_{f}(\mathbb{Z}_{n+m})|$, where
	\begin{equation}
	Stab_{f}(\mathbb{Z}_{n+m}) = \{\sigma \in \mathbb{Z}_{n+m}\mid \sigma f = f\}.
	\end{equation}
	Therefore, for any $[f]\in \tilde{V}_{n,m}$, the coefficient $c_{[f]}$ can be written as
	\begin{equation}
	c_{[f]} = \frac{m}{|Stab_{f}(\mathbb{Z}_{n+m})|}.
	\end{equation}
	And we can write $O_{(n,m)}$ as
	\begin{equation}
	O_{(n,m)} = \sum_{[f] \in \tilde{V}_{(n,m)}}\frac{m}{|Stab_{f}(\mathbb{Z}_{n+m})|} O_{[f]}.
	\end{equation}
	
	For the case $m = 0$, $O_{(n,0)} = \text{tr}(Z^n)$ by definition.

	Now we consider the symmetrized operator $A_{(n,m)}$, which can be written in the form
	\begin{equation}
	A_{(n,m)}  = \sum_{[f] \in \tilde{V}_{n,m}}c'_{[f]} O_{[f]}.
	\end{equation}
We can also write the definition of $A_{(n,m)}$
\begin{equation}
A_{(n,m)} = \frac{1}{(n+m)!}\sum_{\sigma \in S_{n+m}} \sigma O_{[f_0]}.
\end{equation}
For any $f$, we can always find a $\sigma_f$ such that $\sigma_f f_0 = f$. Define $I \subset \{1,2,\dots,n+m\}$ by
\begin{equation}
I = \{i \mid f(i) = \uparrow\}.
\end{equation}
And similarly define  $J \subset \{1,2,\dots,n+m\}$ by
\begin{equation}
J =\{j\mid f(j) = \downarrow\}.
\end{equation}
Then we define $S_I\subset S_{n+m}$ to be the group of permutation of indices in $I$ and similarly for $S_J$. Then for any $\pi \in S_{I}\times S_J$, we have
\begin{equation}
\pi f = f, \quad  \Rightarrow \quad (\pi\sigma_f)f_0 = f.
\end{equation}
Therefore each $f$ is counted $|S_I|\times|S_J|$ times. Then consider the equivalence class $[f]$. We need to count the number of elements in the equivalence class, which is given by $(n+m)/|Stab_{f}(\mathbb{Z}_{n+m})|$. The coefficient $c'_{[f]}$ is then the product of the overall $\frac{1}{(n+m)!}$ with the above two factors, which is
\begin{equation}
c'_{[f]} = \frac{1}{(n+m)!}|S_I|\times|S_J|\frac{n+m}{|Stab_{f}(\mathbb{Z}_{n+m})|} =\frac{n!m!}{(n+m)!} \frac{(n+m)}{|Stab_{f}(\mathbb{Z}_{n+m})|}.
\end{equation}
And we have
\begin{equation}
A_{(n,m)} = \sum_{[f] \in \tilde{V}_{n,m}} \frac{n!m!}{(n+m)!} \frac{(n+m)}{|Stab_{f}(\mathbb{Z}_{n+m})|} O_{[f]}.
\label{Symmetrized_expression}
\end{equation}
Therefore, for $m > 0$, we have
\begin{equation}
A_{(n,m)} =  \frac{n!m!}{(n+m)!} \frac{(n+m)}{m} O_{(n,m)}.
\end{equation}
Here, the formula is not symmetric in $n$ and $m$ since the definition of $O_{(n,m)}$ is not symmetric in $n$ and $m$.

And for $m = 0$, $A_{(n,0)} = \text{tr}(Z^n) = O_{(n,0)}$.
\end{proof}

From the calculation in the above proof, or specifically from the equation (%
\ref{Symmetrized_expression}), we can easily express $a_{(n,m)}^{\dagger }$
by linear combination of $a_{[f]}^{\dagger }$ as follows:
\begin{equation}
a_{(n,m)}^{\dagger }=\sum_{[f]\in \tilde{V}_{n,m}}\frac{n!m!}{(n+m)!}\frac{%
(n+m)}{|Stab_{f}(\mathbb{Z}_{n+m})|}a_{[f]}^{\dagger }.
\end{equation}%
Then we can also write down the commutation relations
\begin{equation}
\lbrack a_{(n,m)},a_{(n^{\prime },m^{\prime })}^{\dagger }]=\delta
_{nn^{\prime }}\delta _{mm^{\prime }}\sum_{[f]\in \tilde{V}_{n,m}}\left(
\frac{n!m!}{(n+m)!}\frac{(n+m)}{|Stab_{f}(\mathbb{Z}_{n+m})|}\right)
^{2}\quad \text{ for }n,m\geq 1
\end{equation}%
and
\begin{equation}
\lbrack a_{(n,0)},a_{(n^{\prime },0)}^{\dagger }]=n\delta _{nn^{\prime
}},\quad \lbrack a_{(0,m)},a_{(0,m^{\prime })}^{\dagger }]=m\delta
_{mm^{\prime }}.
\end{equation}

It's also easy to see that the symmetrized trace operator $A_{(n,m)}$ is in
the kernel of $\frac{1}{2}\hat{D}_{2}^{\prime }$. Note that%
\begin{equation}
(\mathbb{I}-P_{i,i+1})\sum_{\sigma \in S_{n+m}}\sigma =\sum_{\sigma \in
S_{n+m}}\sigma -\sum_{\sigma \in S_{n+m}}P_{i,i+1}\sigma =\sum_{\sigma \in
S_{n+m}}\sigma -\sum_{\sigma \in S_{n+m}}\sigma =0.
\end{equation}%
Therefore $(\mathbb{I}-P_{i,i+1})A_{(n,m)}=0$ and we have
\begin{equation}
\frac{1}{2}\hat{D}_{2}^{\prime }A_{(n,m)}=0.
\end{equation}

\begin{prop}
The quarter BPS coherent states
\begin{equation}
\rvert Coh^{\frac{1}{4}}\rangle =\exp (\sum_{n,m\geq
0}c_{(n,m)}a_{(n,m)}^{\dagger })\rvert 0\rangle
\end{equation}%
where $c_{(n,m)}$ are complex coefficients, are annihilated by $\hat{D}%
_{2}^{\prime }$ and are eigenstates of the annihilation operators $a_{(n,m)}$%
.
\end{prop}

\begin{proof}
The quarter BPS coherent states are
\begin{equation}
\rvert Coh^{\frac{1}{4}}\rangle =\exp (\sum_{n,m\geq
0}c_{(n,m)}a_{(n,m)}^{\dagger })\rvert 0\rangle
\end{equation}%
where $c_{(n,m)}$ are complex coefficients. We can then easily prove that exponential of the symmetrized operator is
still in the kernel of $\hat{D}_{2}^{\prime }$
\begin{eqnarray}
\frac{1}{2}\hat{D}_{2}^{\prime }\exp (\sum_{n,m\geq
0}c_{(n,m)}a_{(n,m)}^{\dagger })\rvert 0\rangle &=&\frac{1}{2}\hat{D}_{2}^{\prime }\prod_{n,m\geq 0}\sum_{l\geq 0}\frac{1}{l!%
}(c_{(n,m)}a_{(n,m)}^{\dagger })^{l}\rvert 0\rangle   \notag \\
&=&\frac{1}{2}\hat{D}_{2}^{\prime }\sum_{l_{(n,m)}}\prod_{n,m\geq 0}\frac{1}{%
l_{(n,m)}!}(c_{(n,m)}a_{(n,m)}^{\dagger })^{l_{(n,m)}}\rvert 0\rangle
\notag \\
&=&\sum_{l_{(n,m)}}\frac{1}{2}\hat{D}_{2}^{\prime }\prod_{n,m\geq 0}\frac{1}{%
l_{(n,m)}!}(c_{(n,m)}a_{(n,m)}^{\dagger })^{l_{(n,m)}}\rvert 0\rangle =0,\notag \\
\end{eqnarray}%
where in the last line we have used the fact that $\frac{1}{2}\hat{D}%
_{2}^{\prime }a_{(n,m)}^{\dagger }=0$ and the action of $\frac{1}{2}\hat{D}%
_{2}^{\prime }$ satisfies Leibniz rule.

Now we show that the state defined above is the eigenstate
of the annihilation operators. First, consider the case when $n,m \geq 1$
\begin{equation}
a_{(n,m)} \exp(\sum_{n,m \geq 0} c_{(n,m)} a^\dagger_{(n,m)}) \rvert 0
\rangle \\
=\left( \sum_{[f] \in \tilde{V}_{n,m}} \frac{n!m!}{(n+m)!} \frac{(n+m)}{%
|Stab_{f}(\mathbb{Z}_{n+m})|} a_{[f]} \right) \exp(\sum_{n,m \geq 0}
c_{(n,m)} a^\dagger_{(n,m)}) \rvert 0 \rangle.
\end{equation}
And for $[f]\in \tilde{V}_{(n,m)}$
\begin{eqnarray}
&&a_{[f]}\exp (\sum_{n,m\geq 0}c_{(n,m)}a_{(n,m)}^{\dagger })\rvert 0\rangle
\\
&=&a_{[f]}\exp (\sum_{n,m\geq 0}\sum_{[f]\in \tilde{V}_{n,m}}c_{(n,m)}\frac{%
n!m!}{(n+m)!}\frac{(n+m)}{|Stab_{f}(\mathbb{Z}_{n+m})|}a_{[f]}^{\dagger
})\rvert 0\rangle \\
&=&\frac{n!m!}{(n+m)!}\frac{(n+m)}{|Stab_{f}(\mathbb{Z}_{n+m})|}%
c_{(n,m)}\exp (\sum_{n,m\geq 0}c_{(n,m)}a_{(n,m)}^{\dagger })\rvert 0\rangle.
\end{eqnarray}
Combining the above two formulas, we have
\begin{equation}
a_{(n,m)}\exp(\sum_{n,m \geq 0} c_{(n,m)} a^\dagger_{(n,m)}) \rvert 0
\rangle = c_{(n,m)}\sum_{[f] \in \tilde{V}_{n,m}} \left( \frac{n!m!}{(n+m)!}
\frac{(n+m)}{|Stab_{f}(\mathbb{Z}_{n+m})|} \right)^2 \exp(\sum_{n,m \geq 0}
c_{(n,m)} a^\dagger_{(n,m)}) \rvert 0 \rangle.
\end{equation}

Then we consider the case $m=0$. In this
situation, $V_{n,0}$ contains only one element $f=\uparrow ^{n}$, and $%
|Stab_{\uparrow ^{n}}(\mathbb{Z}_{n})|=n$. Then we have $a_{(n,0)}^{\dagger
}=a_{n}^{\dagger }$ which corresponds to our previous half BPS operator. The derivation is formally the same.
\begin{eqnarray}
a_{(n,0)}\exp (\sum_{n,m\geq 0}c_{(n,m)}a_{(n,m)}^{\dagger })\rvert 0\rangle
&=&\exp (\sum_{n\geq 1,m\geq 0}c_{(n,m)}a_{(n,m)}^{\dagger })a_{n}\exp
(\sum_{k}c_{(k,0)}a_{k}^{\dagger })\rvert 0\rangle  \notag \\
&=&\exp (\sum_{n\geq 1,m\geq 0}c_{(n,m)}a_{(n,m)}^{\dagger })nc_{(n,0)}\exp
(\sum_{k}c_{(k,0)}a_{k}^{\dagger })\rvert 0\rangle  \notag \\
&=&nc_{(n,0)}\exp (\sum_{n,m\geq 0}c_{(n,m)}a_{(n,m)}^{\dagger })\rvert
0\rangle.
\end{eqnarray}
Similarly, for $n=0$, we have
\begin{equation}
a_{(0,m)}\exp (\sum_{n,m\geq 0}c_{(n,m)}a_{(n,m)}^{\dagger })\rvert 0\rangle
=mc_{(0,m)}\exp (\sum_{n,m\geq 0}c_{(n,m)}a_{(n,m)}^{\dagger })\rvert
0\rangle .
\end{equation}

\end{proof}

The stabilizer subgroup $Stab_{f}(\mathbb{Z}_{n+m})$ plays an important role
in our above results. Basically, it determines the normalization of the
symmetrized operator $A_{(n,m)}$. It's easy to see that they satisfy the
following formula
\begin{equation}
\sum_{\lbrack f]\in \tilde{V}_{n,m}}\frac{n!m!}{(n+m)!}\frac{(n+m)}{%
|Stab_{f}(\mathbb{Z}_{n+m})|}=1.
\end{equation}

\vspace{1pt}

\vspace{1pt}

\subsection{Inner products with quarter BPS coherent states}

The truncated coherent state is
\begin{equation}
\rvert Coh^{\frac{1}{4}}(x,y,c)\rangle =\exp (\sum_{n}\frac{\Lambda _{n}}{n}%
a_{(n,0)}^{\dagger }+\sum_{m}\frac{\Lambda _{m}^{\prime }}{m}%
a_{(0,m)}^{\dagger }+ca_{(1,1)}^{\dagger })\rvert 0\rangle ,
\end{equation}%
with
\begin{equation}
\Lambda _{n}=x_{1}^{n}+x_{2}^{n}+\cdots ,\quad \Lambda _{m}^{\prime
}=y_{1}^{m}+y_{2}^{m}+\cdots .
\end{equation}%
Using equation (\ref{Inner_general_Coh}) for the overlap of coherent states
in a general form, we can directly write the overlap between two truncated
coherent states as follows:
\begin{equation}
\langle Coh^{\frac{1}{4}}(x,y,c)\rvert Coh^{\frac{1}{4}}(x^{\prime
},y^{\prime },c^{\prime })\rangle =\prod_{i,j}\frac{1}{1-x_{i}^{\prime
}x_{j}^{\ast }}\prod_{i,j}\frac{1}{1-y_{i}^{\prime }y_{j}^{\ast }}\exp
(c^{\prime }c^{\ast }).
\end{equation}%
Then for the special case where $x^{\prime }=x,y^{\prime }=y,c^{\prime }=c$,
we get the norm-squared of a truncated coherent state:
\begin{equation}
\langle Coh^{\frac{1}{4}}(x,y,c)\rvert Coh^{\frac{1}{4}}(x,y,c)\rangle
=\prod_{i,j}\frac{1}{1-x_{i}x_{j}^{\ast }}\prod_{i,j}\frac{1}{%
1-y_{i}y_{j}^{\ast }}\exp (|c|^{2}).
\end{equation}

We consider the overlap of the coherent states with trace product basis.
Using notation $a_{[f]}^{\dagger }$ in this section, then a general basis
can be written as
\begin{equation}
\prod_{\lbrack f]\in \tilde{V}}a_{[f]}^{\dagger w_{[f]}}\rvert 0\rangle ,
\end{equation}%
where $\tilde{V}=\bigcup_{n,m}\tilde{V}_{n,m}$.

Then we consider
\begin{equation}
\langle 0\lvert \prod_{\lbrack f]\in \tilde{V}}a_{[f]}^{w_{[f]}}\rvert Coh^{%
\frac{1}{4}}\rangle ,
\end{equation}%
where $\rvert Coh^{\frac{1}{4}}\rangle =\exp
(\sum_{n,m}c_{(n,m)}a_{(n,m)}^{\dagger })\rvert 0\rangle $. Since $\rvert
Coh^{\frac{1}{4}}\rangle $ is just a special case of our previous $\rvert
Coh\rangle =\exp (\sum_{[f]\in \tilde{V}}c_{[f]}a_{[f]}^{\dagger })\rvert
0\rangle $, we still use the equation (\ref{coh_trace_product_02}). And we
have that
\begin{equation}
\langle 0\lvert \prod_{\lbrack f]\in \tilde{V}}a_{[f]}^{w_{[f]}}\rvert
Coh\rangle =\prod_{n}(nc_{[\uparrow ^{n}]})^{w_{[\uparrow
^{n}]}}\prod_{m}(mc_{[\downarrow ^{m}]})^{w_{[\downarrow
^{m}]}}\prod_{[f]\in \cup _{n,m\geq 1}\tilde{V}_{n,m}}(c_{[f]})^{w_{[f]}},
\label{Overlap_general}
\end{equation}%
where $\uparrow ^{n}$ is defined to map every $i\in \{1,\dots ,n\}$ to $%
\uparrow $, and similarly for $\downarrow ^{m}$. We can write $\rvert Coh^{%
\frac{1}{4}}\rangle $ as
\begin{equation}
\rvert Coh^{\frac{1}{4}}\rangle =\exp (\sum_{n,m\geq 0}\sum_{[f]\in \tilde{V}%
_{n,m}}c_{(n,m)}\frac{n!m!}{(n+m)!}\frac{(n+m)}{|Stab_{f}(\mathbb{Z}_{n+m})|}%
a_{[f]}^{\dagger })\rvert 0\rangle .
\end{equation}%
Therefore, replacing $c_{[f]}$ by $c_{(n,m)}\frac{n!m!}{(n+m)!}\frac{(n+m)}{%
|Stab_{f}(\mathbb{Z}_{n+m})|}$ in (\ref{Overlap_general}), we have
\begin{eqnarray}
&&\langle 0\lvert \prod_{\lbrack f]\in \tilde{V}}a_{[f]}^{w_{[f]}}\rvert
Coh^{\frac{1}{4}}\rangle  \notag \\
&=&\prod_{n}(nc_{(n,0)})^{w_{[\uparrow
^{n}]}}\prod_{m}(mc_{(0,m)})^{w_{[\downarrow ^{m}]}}\prod_{[f]\in \cup
_{n,m\geq 1}\tilde{V}_{n,m}}\left( c_{(n,m)}\frac{n!m!}{(n+m)!}\frac{(n+m)}{%
|Stab_{f}(\mathbb{Z}_{n+m})|}\right) ^{w_{[f]}}.  \notag
\label{Overlap_general_02} \\
&&
\end{eqnarray}%
In some situations, the above formula will give zero, and whether the above
overlap equals zero depends on which $c_{(n,m)}$ equals zero.

For the special case of half BPS states, we have
\begin{equation}
w_{[f]}=0\text{ for }[f]\notin \tilde{V}_{n,0},
\end{equation}%
and the coherent state is
\begin{equation}
\rvert Coh^{\frac{1}{2}}\rangle =\exp (\sum_{n}c_{(n,0)}a_{[\uparrow
^{n}]}^{\dagger })\rvert 0\rangle .
\end{equation}%
We have the overlap
\begin{equation}
\langle 0\lvert \prod_{n}a_{[\uparrow ^{n}]}^{w_{[\uparrow ^{n}]}}\rvert
Coh^{\frac{1}{2}}\rangle =\prod_{n}(nc_{(n,0)})^{w_{[\uparrow ^{n}]}}.
\end{equation}%
Using notation in \cite{Lin:2017dnz}
\begin{equation}
a_{[\uparrow ^{n}]}^{\dagger }=a_{n}^{\dagger },\;w_{[\uparrow
^{n}]}=w_{n},\;c_{(n,0)}=\frac{\Lambda _{n}}{n},
\end{equation}%
this reduces to our previous results in \cite{Lin:2017dnz},
\begin{equation}
\langle 0\lvert \prod_{n}a_{n}^{w_{n}}\rvert Coh^{\frac{1}{2}}\rangle
=\prod_{n}(\Lambda _{n})^{w_{n}}.
\end{equation}

Till now, we have established the basic ingredients related to the quarter
BPS coherent states and studied their properties including inner product. In
the following section we will consider other interesting operators and look
at their relations.

\vspace{1pt}

\section{Brauer states and their relation to quarter BPS coherent states}

\renewcommand{\theequation}{3.\arabic{equation}} \setcounter{equation}{0} %
\renewcommand{\thethm}{3.\arabic{thm}} \setcounter{thm}{0} %
\renewcommand{\theprop}{3.\arabic{prop}} \setcounter{prop}{0} %
\renewcommand{\thelemma}{3.\arabic{lemma}} \setcounter{lemma}{0}

\subsection{Brauer operators and trace product basis}

The Brauer algebra \cite{Brauer:1937} is a natural generalization of the
symmetric group algebra. The Walled Brauer algebra \cite%
{Benkart:1994,Koike:1989,Turaev:1989} is a subalgebra of Brauer algebra.
Kimura and Ramgoolam constructed basis of gauge invariant operators of two
matrix fields, labelled by representations of Brauer algebras \cite%
{Kimura:2007wy,Kimura:2010tx}. The Walled Brauer algebra $B_{N}(n,m)$ is
very natural to multi-matrix case. The Schur-Weyl duality for Walled Brauer
algebra is
\begin{equation}
V^{\otimes n}\bigotimes \bar{V}^{\otimes m}=\bigoplus_{\gamma }V_{\gamma
}^{U(N)}\bigotimes V_{\gamma }^{B_{N}(n,m)}
\end{equation}%
with $V$ and $\bar{V}$ corresponding to the two matrices. In the above
formula, $\gamma $ runs over the set of all staircases. A staircase is
defined to be a sequence of integers $(\gamma _{1},\gamma _{2},\dots ,\gamma
_{r})$ such that $\gamma _{1}\geq \gamma _{2}\geq \dots \geq \gamma _{r}$.
The sets of positive integers and negative integers determine two partitions
respectively, hence we can equivalently write a staircase as $\gamma
=(l,\gamma ^{+},\gamma ^{-})$ where $l$ is lying in between $0$ and $\min
(n,m)$. And $\gamma ^{+},\gamma ^{-}$ are two Young tableaux. For more
detail of the Walled Brauer algebra, see \cite{Kimura:2007wy,Kimura:2010tx}.

We will consider $b$ to be an element of a basis of the Brauer algebra $%
B_{N}(n,m)$. Although the basis can be more general in our following
definitions, we will consider the basis to be the set of $(n,m)$ diagrams $%
\mathcal{D}_{n,m}$ defined in Appendix \ref{appendix brauer algebra}. The
dimension of the $U(N)$ irreducible representation associated with the label
$\gamma $, is denoted as $t_{\gamma }=\dim \gamma $. And the dual element $%
b^{\ast }$ is defined and computed in \cite{Kimura:2007wy}, which is as
follows
\begin{equation}
b^{\ast }=\frac{1}{N^{n+m}}\Sigma ^{-1}(\Omega _{n+m}^{-1}(\Sigma (b))^{-1})
\end{equation}%
where the map $\Sigma :B_{N}(n,m)\rightarrow \mathbb{C}[S_{n+m}]$ is defined
in (3.19) in \cite{Kimura:2007wy}. The operator $\Omega _{n+m}$ is defined
by $\Omega _{n+m}=\sum_{\sigma \in S_{n+m}}N^{C(\sigma )-(n+m)}\sigma $. And
its inverse is
\begin{equation}
\Omega _{n+m}^{-1}=\frac{N^{n+m}}{(n+m)!}\sum_{T}\frac{d_{T}^{2}}{\dim T}%
\sum_{\sigma \in S_{n+m}}\chi _{T}(\sigma )\sigma .
\end{equation}

There are projection operators $P^{\gamma }$ which can be expressed as
\begin{equation}
P^{\gamma }=t_{\gamma }\sum_{b}\chi ^{\gamma }(b)b^{\ast }=t_{\gamma
}\sum_{b}\chi _{\gamma }(b^{\ast })b,  \label{Projector P}
\end{equation}%
where $\chi ^{\gamma }$ is the irreducible representation labeled by
staircase $\gamma $.

The irreducible representation $V_{\gamma }^{B_{N}(n,m)}$ of $B_{N}(n,m)$
can be further decomposed into a number of representations of the subalgebra
$\mathbb{C}[S_{n}\times S_{m}]$. We can introduce the operators $%
Q_{A,i,j}^{\gamma }$ as follows:
\begin{equation}
Q_{A,ij}^{\gamma }=t_{\gamma }\sum_{b}\chi _{A,ij}^{\gamma }(b^{\ast })b.
\end{equation}%
Here $A$ labels irreducible representations of $\mathbb{C}[S_{n}\times
S_{m}] $ and $i,j$ run over the multiplicity of the relevant decomposition.
The $\chi _{A,ij}^{\gamma }$ is the restricted character. For more detail,
see \cite{Kimura:2008ac,Kimura:2010tx}. We see that $P^{\gamma
}=\sum_{A,i}Q_{A,ii}^{\gamma }$. The operators $P^{\gamma },Q_{A,ij}^{\gamma
}$ have the property
\begin{equation}
hP^{\gamma }h^{-1}=P^{\gamma },~hQ_{A,ij}^{\gamma }h^{-1}=Q_{A,ij}^{\gamma
},\quad \text{ for all }h\in S_{n}\times S_{m}\subset B_{N}(n,m).
\end{equation}%
By the equation (3.37) in \cite{Kimura:2007wy}, we have
\begin{equation}
h\Sigma (P^{\gamma })h^{-1}=\Sigma (P^{\gamma }),~h\Sigma (Q_{A,ij}^{\gamma
})h^{-1}=\Sigma (Q_{A,ij}^{\gamma }),\quad \text{ for all }h\in S_{n}\times
S_{m}\subset B_{N}(n,m).
\end{equation}%
Here, the inclusion $S_{n}\times S_{m}\subset B_{N}(n,m)~$means that we can
take from the algebra $B_{N}(n,m)$ a subset $S_{n}\times S_{m}$ which is
also a group. The Brauer operators take the form
\begin{equation}
O_{A,ij}^{\gamma }(Z,Y)=\text{tr}(Q_{A,ij}^{\gamma }Z^{\otimes n}Y^{T\otimes
m}).  \label{Brauer_01}
\end{equation}%
There are also operators
\begin{equation}
O^{\gamma }(Z,Y)=\text{tr}(P^{\gamma }Z^{\otimes n}Y^{T\otimes m}).
\label{Brauer_02}
\end{equation}%
The inner products of the Brauer operators are described in Appendix \ref%
{Appendix A}.

Properties of the map $\Sigma $ is needed to compute various quantities, and
an important formula is (3.36) in \cite{Kimura:2007wy}, which is
\begin{equation}
\text{tr}(\Sigma (b)Z^{\otimes n}Y^{\otimes m})=\text{tr}(bZ^{\otimes
n}Y^{T\otimes m}).  \label{Sigma map}
\end{equation}

We will consider the inner product between $O_{A,ij}^{\gamma }(Z,Y)$ and our
previous trace product basis. Using (\ref{Projector P}-\ref{Sigma map}), the
inner product can be computed%
\begin{eqnarray}
\langle (\text{tr}(\alpha Z^{\otimes n}\otimes Y^{\otimes m}))^{\dagger
}O_{A,ij}^{\gamma }(Z,Y)\rangle &=&\langle (\text{tr}(\alpha Z^{\otimes
n}\otimes Y^{\otimes m}))^{\dagger }\text{tr}(Q_{A,ij}^{\gamma }Z^{\otimes
n}Y^{T\otimes m})\rangle  \notag \\
&=&t_{\gamma }\sum_{b}\chi _{\gamma }(Q_{A,ij}^{\gamma }b^{\ast })\langle (%
\text{tr}(\alpha Z^{\otimes n}\otimes Y^{\otimes m}))^{\dagger }\text{tr}%
(bZ^{\otimes n}Y^{T\otimes m})\rangle  \notag \\
&=&t_{\gamma }\sum_{b}\chi _{\gamma }(Q_{A,ij}^{\gamma }b^{\ast })\langle (%
\text{tr}(\alpha Z^{\otimes n}\otimes Y^{\otimes m}))^{\dagger }\text{tr}%
(\Sigma (b)Z^{\otimes n}Y^{\otimes m})\rangle  \notag \\
&=&t_{\gamma }\sum_{b}\chi _{\gamma }(Q_{A,ij}^{\gamma }b^{\ast })\sum_{h\in
S_{n}\times S_{m}}\text{tr}(\alpha ^{-1}h^{-1}\Sigma (b)h)  \notag \\
&=&\sum_{h\in S_{n}\times S_{m}}\text{tr}(\alpha ^{-1}h^{-1}\Sigma
(Q_{A,ij}^{\gamma })h)  \notag \\
&=&n!m!\text{tr}(\alpha ^{-1}\Sigma (Q_{A,ij}^{\gamma })).
\end{eqnarray}%
Similarly, we have formula for the operators $O^{\gamma }(Z,Y)$
\begin{equation}
\langle (\text{tr}(\alpha Z^{\otimes n}\otimes Y^{\otimes m}))^{\dagger
}O^{\gamma }(Z,Y)\rangle =n!m!\text{tr}(\alpha ^{-1}\Sigma (P^{\gamma })).
\label{Brauer_trace_product_02}
\end{equation}

Our next goal is to find more explicit expressions of our above results.
However, $t_{\gamma }$, $\chi _{\gamma }$, $b^{\ast }$ all depend on $N$,
therefore we need more formulas for the $N$ dependence of these quantities.

The $t_{\gamma }=\dim \gamma $ is the dimension of the $U(N)$ irreducible
representation associated with the label $\gamma $. The dimension can be
computed as follows. First $\gamma =(l,\gamma ^{+},\gamma ^{-})$ can always
be represented as a sequence $(\gamma _{1},\gamma _{2},\dots ,\gamma _{N})$
with $\gamma _{1}\geq \gamma _{2}\geq \dots \geq \gamma _{N}$. We define $%
\lambda (\gamma )=(\lambda _{1},\dots ,\lambda _{N})$ with $\lambda
_{i}=\gamma _{i}-\gamma _{N}+1$. By definition, $\lambda (\gamma )$ is a
partition, therefore it corresponds to an irreducible representation of $%
U(N) $ and we know the dimension formula for this representation, and we
define%
\begin{equation}
\dim \gamma =\dim \lambda (\gamma )=\prod_{(i,j)\in \lambda (\gamma )}\frac{%
N-i+j}{h_{ij}}.
\end{equation}%
We have that $\dim \gamma =O(N^{n+m-2l})$.\vspace{1pt}$~$It's interesting to
know the leading order behavior of $t_{\gamma }$. Therefore we define $%
\tilde{t}_{\gamma }$ as%
\begin{equation}
\tilde{t}_{\gamma }=\lim_{N\rightarrow \infty }\frac{1}{N^{n+m-2l}}t_{\gamma
},
\end{equation}%
where $\gamma =(l,\gamma ^{+},\gamma ^{-})$. We will give an explicit
expression for the coefficient $\tilde{t}_{\gamma }$.

From the above definition, we see that $\dim \gamma $ does not explicitly
depend on $l$. The $\dim \gamma $ only depends on $l$ through its dependence
on $\gamma ^{+},\gamma ^{-}$, or in other words, $\dim \gamma $ is a
function $f(N,\gamma ^{+},\gamma ^{-})$ of variables $N,\gamma ^{+},\gamma
^{-}$. Then we rewrite the expression of $\tilde{t}_{\gamma }$:
\begin{equation}
\tilde{t}_{\gamma }=\lim_{N\rightarrow \infty }\frac{1}{N^{|\gamma
^{+}|+|\gamma ^{-}|}}f(N,\gamma ^{+},\gamma ^{-}).
\end{equation}

We find that $\tilde{t}_{\gamma }$ only depends on $\gamma ^{+},\gamma ^{-}$,%
\begin{equation}
\tilde{t}_{\gamma }=\tilde{t}(\gamma ^{+},\gamma ^{-}).
\end{equation}%
\vspace{1pt}We find a formula to calculate the coefficient $\tilde{t}%
_{\gamma }$ explicitly in the case $\gamma =(0,\gamma ^{+},\gamma ^{-})$.
First, using equation (4.18) in \cite{Kimura:2007wy}:
\begin{equation}
\frac{d_{R}^{2}d_{S}^{2}}{\dim R\bar{S}}=\frac{m!^{2}n!^{2}}{(m+n)!}%
\sum_{T\vdash (m+n)}\frac{d_{T}^{2}}{\dim T}g(R,S;T)
\end{equation}%
where in this notation, $\gamma =(0,R,S)$, and $\dim R\bar{S}=\dim \gamma $.
The $d_{R}$ for a Young tableau $R\vdash n$ is the dimension of the
irreducible representation of $S_{n}$ labeled by $R$. The $\dim T$ for $%
T\vdash (n+m)$ is the dimension of the irreducible representation of $U(N)$
labeled by $T$. In our notation, we have
\begin{equation}
\frac{d_{\gamma ^{+}}^{2}d_{\gamma ^{-}}^{2}}{\dim \gamma }=\frac{%
m!^{2}n!^{2}}{(m+n)!^{2}}\sum_{T\vdash (m+n)}\frac{d_{T}^{2}}{\dim T}%
g(\gamma ^{+},\gamma ^{-};T).
\end{equation}%
And we use formulas%
\begin{equation}
\dim T=\prod_{(i,j)\in T}\frac{N-i+j}{h_{ij}},\;\quad
d_{T}=(n+m)!\prod_{(i,j)\in T}\frac{1}{h_{ij}}.
\end{equation}%
This gives us%
\begin{eqnarray}
d_{\gamma ^{+}}^{2}d_{\gamma ^{-}}^{2} &=&\frac{m!^{2}n!^{2}}{(m+n)!}%
\sum_{T\vdash (m+n)}\dim \gamma \frac{d_{T}^{2}}{\dim T}g(\gamma ^{+},\gamma
^{-};T)  \notag \\
&=&\frac{m!^{2}n!^{2}}{(m+n)!^{2}}\sum_{T\vdash (m+n)}\dim \gamma d_{T}\frac{%
(n+m)!}{\prod_{(i,j)\in T}(N-i+j)}g(\gamma ^{+},\gamma ^{-};T)  \notag \\
&=&\lim_{N\rightarrow \infty }\frac{m!^{2}n!^{2}}{(m+n)!^{2}}\sum_{T\vdash
(m+n)}d_{T}\dim \gamma \frac{(n+m)!}{N^{n+m}}(1+O(1/N))g(\gamma ^{+},\gamma
^{-};T)  \notag \\
&=&\frac{m!^{2}n!^{2}}{(m+n)!}\tilde{t}_{\gamma }\sum_{T\vdash
(m+n)}d_{T}g(\gamma ^{+},\gamma ^{-};T).
\end{eqnarray}%
Then we use equation (4.10) in \cite{Kimura:2007wy}:
\begin{equation}
\frac{(m+n)!}{m!n!}d_{R}d_{S}=\sum_{T}g(R,S;T)d_{T}.
\end{equation}%
We find that
\begin{equation}
\tilde{t}_{(0,\gamma ^{+},\gamma ^{-})}=\frac{d_{\gamma ^{+}}d_{\gamma ^{-}}%
}{n!m!}=\frac{d_{\gamma ^{+}}d_{\gamma ^{-}}}{|\gamma ^{+}|!|\gamma ^{-}|!}.
\end{equation}%
\vspace{1pt}From this we have%
\begin{equation}
\tilde{t}_{\gamma }=\tilde{t}(\gamma ^{+},\gamma ^{-})=\frac{d_{\gamma
^{+}}d_{\gamma ^{-}}}{|\gamma ^{+}|!|\gamma ^{-}|!}.
\end{equation}%
For $\gamma =(l,\gamma ^{+},\gamma ^{-})$, $|\gamma ^{+}|=n-l,|\gamma
^{-}|=m-l$, therefore%
\begin{equation}
\tilde{t}_{\gamma }=\tilde{t}_{(l,\gamma ^{+},\gamma ^{-})}=\frac{d_{\gamma
^{+}}d_{\gamma ^{-}}}{(n-l)!(m-l)!}.
\end{equation}

\vspace{1pt}

\subsection{Results for $l=0$ Brauer states}

As a warm up, we first consider the case $l=0$. However, many results
obtained in this subsection can be useful for the analysis of more general $%
l\neq 0$ case. The trace product operator is normalized as $N^{-(n+m)/2}%
\text{tr}(\alpha Z^{\otimes n}\otimes Y^{\otimes m}))$, and the Brauer basis
is normalized as $N^{-(n+m)/2}O^{\gamma }(Z,Y)$. Therefore we should
consider
\begin{equation}
\frac{1}{N^{n+m}}\langle (\text{tr}(\alpha Z^{\otimes n}\otimes Y^{\otimes
m}))^{\dagger }O^{\gamma }(Z,Y)\rangle .
\end{equation}

The projection operator has an expression
\begin{equation}
P^{\gamma }=t_{\gamma }\frac{1}{N^{n+m}}\sum_{b}\chi ^{\gamma }(\Sigma
^{-1}(\Omega _{n+m}^{-1}(\Sigma (b))^{-1}))b.
\end{equation}%
In the above expression, $t_{\gamma },\chi ^{\gamma },\Omega _{n+m}$ all
depend on $N$. And we will analyze the $N$ dependence of each term. The $%
\Omega _{n+m}^{-1}$ has the expression
\begin{equation}
\Omega _{n+m}^{-1}=\frac{N^{n+m}}{(n+m)!^{2}}\sum_{T\vdash (n+m)}\frac{%
d_{T}^{2}}{\dim T}\sum_{\sigma \in S_{n+m}}\chi _{T}(\sigma )\sigma ,
\end{equation}%
where the $\dim T$, $d_{T}$ are%
\begin{equation}
\dim T=\prod_{(i,j)\in T}\frac{N-i+j}{h_{ij}},\;\quad
d_{T}=(n+m)!\prod_{(i,j)\in T}\frac{1}{h_{ij}}.
\end{equation}%
Therefore%
\begin{equation}
\frac{d_{T}}{\dim T}=\frac{(n+m)!}{\prod_{(i,j)\in T}(N-i+j)}=\frac{(n+m)!}{%
N^{n+m}}(1+O(1/N)).
\end{equation}%
Hence we have for $\Omega _{n+m}^{-1}$
\begin{equation}
\Omega _{n+m}^{-1}=\frac{1}{(n+m)!}\sum_{\sigma \in S_{n+m}}\sum_{T\vdash
(n+m)}d_{T}\chi _{T}(\sigma )(1+O(1/N))\sigma .
\end{equation}

If we assume that $N\geq (n+m)$, then the above summation is over all Young
tableaux with $T\vdash (n+m)$. From the representation theory of symmetric
groups \cite{James Kerber,Sagan}, we know that $\sum_{T\vdash
(n+m)}d_{T}\chi _{T}$ is just the character of the regular representation $%
\sum_{T\vdash (n+m)}d_{T}\chi _{T}(\sigma )=\chi _{\text{reg}}(\sigma
)=(n+m)!\delta (\sigma )$, where $\delta (\sigma )=1$ if and only if $\sigma
=1$ and $\delta (\sigma )=0$ otherwise. Therefore we have
\begin{equation}
\Omega _{n+m}^{-1}=(1+O(1/N)).
\end{equation}

According to Theorem 7.20 in \cite{Halverson:1996}, the character $\chi
^{\gamma }$ has expression
\begin{equation}
\chi ^{\gamma }(\zeta )=N^{h}\sum_{\substack{ \lambda \vdash m^{\prime }  \\ %
\pi \vdash n^{\prime }}}\left( \sum_{\delta \vdash (l-h)}g(\delta ,\gamma
^{+};\lambda )g(\delta ,\gamma ^{-};\pi )\right) \chi _{S_{n^{\prime
}}}^{\lambda }(\zeta ^{+})\chi _{S_{m^{\prime }}}^{\pi }(\zeta ^{-})
\label{character Brauer}
\end{equation}%
where $\gamma =(l,\gamma ^{+},\gamma ^{-})$, with $\gamma ^{+}\vdash
n-l,\gamma ^{-}\vdash m-l$. And $\zeta =(h,\zeta ^{+},\zeta ^{-})$, with $%
\zeta ^{+}\vdash n^{\prime }=n-h,\zeta ^{-}\vdash m^{\prime }=m-h$. The $%
\zeta $ is a staircase that represents a character class of $B_{N}(n,m)$.
The character classes of Brauer algebra are reviewed in Appendix \ref%
{appendix brauer algebra}. The coefficient $g(\delta ,\gamma ^{-};\pi )$ is
the Littlewood-Richardson coefficient. The above formula tells us that $\chi
^{\gamma }(\zeta )\neq 0$ only if $h\leq l$.

For $l=0$, the character $\chi ^{\gamma }$ can be further simplified. In
this case, $\chi ^{\gamma }(\zeta )\neq 0$ only if $h=0$, which means that $%
\zeta $ represents an element $b\in S_{n}\times S_{m},$ where $\mathbb{C}%
[S_{n}\times S_{m}]$ $\subset B_{N}(n,m)$, therefore we have

\begin{equation}
\chi ^{\gamma }(b)=%
\begin{cases}
0,~~b\notin S_{n}\times S_{m} \\
\chi _{S_{n}}^{\gamma ^{+}}(b_{1})\chi _{S_{m}}^{\gamma
^{-}}(b_{2}),~~b=b_{1}\otimes b_{2}\in S_{n}\times S_{m}%
\end{cases}%
.
\end{equation}

Putting these together, we have%
\begin{eqnarray}
&&\frac{1}{N^{n+m}}\langle \text{tr}(\alpha Z^{\otimes n}\otimes Y^{\otimes
m})^{\dagger }O^{\gamma }(Z,Y)\rangle  \notag \\
&=&\tilde{t}_{\gamma }n!m!\sum_{b\in S_{n}\times S_{m}}\chi ^{\gamma
}(\Sigma ^{-1}(\Sigma (b)^{-1}))\frac{1}{N^{n+m}}\text{tr}(\alpha
^{-1}\Sigma (b))(1+O(1/N))  \notag \\
&=&\tilde{t}_{\gamma }n!m!\sum_{b_{1}\in S_{n}}\sum_{b_{2}\in S_{m}}\chi
^{\gamma }(b_{1}^{-1}\otimes b_{2}^{-1})\frac{1}{N^{n+m}}\text{tr}(\alpha
^{-1}b_{1}\otimes b_{2}^{-1})(1+O(1/N))  \notag \\
&=&\tilde{t}_{\gamma }n!m!\sum_{b_{1}\in S_{n}}\chi ^{\gamma
^{+}}(b_{1}^{-1})\sum_{b_{2}\in S_{m}}\chi ^{\gamma
^{-}}(b_{2}^{-1})N^{C(\alpha ^{-1}b_{1}\otimes b_{2}^{-1})-(n+m)}(1+O(1/N)).
\notag \\
&&
\end{eqnarray}%
In the above formula, we have used the properties of $\Sigma $ that
\begin{equation}
\Sigma (\sigma \otimes \tau )=\sigma \otimes \tau ^{-1},\quad \text{ for }%
\sigma \otimes \tau \in S_{n}\times S_{m}.
\end{equation}%
This gives us:
\begin{equation}
\frac{1}{N^{n+m}}\langle \text{tr}(\alpha Z^{\otimes n}\otimes Y^{\otimes
m})^{\dagger }O^{\gamma }(Z,Y)\rangle =%
\begin{cases}
O(1/N),~~~~~~\alpha \notin S_{n}\times S_{m} \\
\tilde{t}_{\gamma }n!m!\chi ^{\gamma ^{+}}(\alpha _{1})\chi ^{\gamma
^{-}}(\alpha _{2})(1+O(1/N)),~~~\alpha =\alpha _{1}\otimes \alpha _{2}\in
S_{n}\times S_{m}%
\end{cases}%
\end{equation}%
where we have used the fact that for the symmetric group $S_{n}$ a
permutation $\sigma $ and its inverse $\sigma ^{-1}$ are in the same
conjugacy class, or in other words, $\chi ^{\gamma ^{+}}(\alpha _{1})=\chi
^{\gamma ^{+}}(\alpha _{1}^{-1})$.

We can write the Brauer state by%
\begin{equation}
\rvert \gamma \rangle \leftrightarrow O^{\gamma }(\frac{Z}{\sqrt{N}},\frac{Y%
}{\sqrt{N}})\rvert 0\rangle .
\end{equation}%
Using our previous notation, the trace product state is%
\begin{equation}
\rvert \lbrack \alpha ]\rangle \leftrightarrow \text{tr}(\alpha \left( \frac{%
Z}{\sqrt{N}}\right) ^{\otimes n}\otimes \left( \frac{Y}{\sqrt{N}}\right)
^{\otimes m}).
\end{equation}%
Hence we can equivalently write our above results as%
\begin{eqnarray}
\langle \lbrack \alpha ]\lvert (0,\gamma ^{+},\gamma ^{-})\rangle &=&%
\begin{cases}
O(1/N),~~~~\alpha \neq \alpha _{1}\otimes \alpha _{2} \\
d_{\gamma ^{+}}d_{\gamma ^{-}}\chi ^{\gamma ^{+}}(\alpha _{1})\chi ^{\gamma
^{-}}(\alpha _{2})(1+O(1/N)),~~~~\alpha =\alpha _{1}\otimes \alpha _{2}%
\end{cases}
\\
\lvert (0,\gamma ^{+},\gamma ^{-})\rangle &=&d_{\gamma ^{+}}d_{\gamma
^{-}}\rvert \gamma ^{+}\rangle \otimes \rvert \gamma ^{-}\rangle +O(1/N).
\end{eqnarray}

Our above derivation is performed for $l=0$, but the analysis of $t_{\gamma
} $ and $\Omega _{n+m}$ is useful for general case. To consider the
situation for $l\neq 0$, we only need to figure out the behavior of $\chi
^{\gamma }$, which is analyzed in the next subsection.

\subsection{Results for $l\neq 0$ Brauer states}

In the $l\neq 0$ case, we need to consider a normalization for the Brauer
state. We take $\rvert (l,\gamma ^{+},\gamma ^{-})\rangle \leftrightarrow
N^{l}O^{\gamma }(\frac{Z}{\sqrt{N}},\frac{Y}{\sqrt{N}})$. The meaning of the
factor $N^{l}$ will be clear when we later discuss the character of Brauer
algebra. Then we consider the overlap
\begin{equation}
\frac{1}{N^{n+m}}\langle \text{tr}(\alpha Z^{\otimes n}\otimes Y^{\otimes
m})^{\dagger }N^{l}O^{\gamma }(Z,Y)\rangle .
\end{equation}%
And we already have the expression (\ref{Brauer_trace_product_02})%
\begin{equation}
\langle \text{tr}(\alpha Z^{\otimes n}\otimes Y^{\otimes m})^{\dagger
}O^{\gamma }(Z,Y)\rangle =n!m!\text{tr}(\alpha ^{-1}\Sigma (P^{\gamma })),
\end{equation}%
where $P^{\gamma }=t_{\gamma }\frac{1}{N^{n+m}}\sum_{b}\chi ^{\gamma
}(\Sigma ^{-1}(\Omega _{n+m}^{-1}(\Sigma (b))^{-1}))b$. We have leading
order results for $\Omega _{n+m}^{-1}$ and $t_{\gamma }$,%
\begin{equation}
\Omega _{n+m}^{-1}=1+O(1/N),~~~~~~t_{\gamma }=\tilde{t}_{\gamma
}N^{n+m-2l}(1+O(1/N)).
\end{equation}%
Therefore we have%
\begin{eqnarray}
&&\frac{1}{N^{n+m}}\langle \text{tr}(\alpha Z^{\otimes n}\otimes Y^{\otimes
m})^{\dagger }N^{l}O^{\gamma }(Z,Y)\rangle  \notag \\
&=&n!m!\tilde{t}_{\gamma }\sum_{b}\frac{1}{N^{l}}\chi ^{\gamma }(\Sigma
^{-1}(\Sigma (b)^{-1}))\frac{\text{tr}(\alpha ^{-1}\Sigma (b))}{N^{n+m}}%
(1+O(1/N))  \notag \\
&=&n!m!\tilde{t}_{\gamma }\frac{1}{N^{l}}\chi ^{\gamma }(\Sigma ^{-1}(\alpha
^{-1}))(1+O(1/N))
\end{eqnarray}%
where in the third line we have used that $\text{tr}(\alpha ^{-1}\Sigma
(b))=N^{n+m}$ only for $\Sigma (b)=\alpha $. Write $b_{\alpha }=\Sigma
^{-1}(\alpha ^{-1})$, we need to find the order $N^{l}$ value of $\chi
^{\gamma }(b_{\alpha })$.

In Appendix \ref{appendix brauer algebra}, we give a rather detailed
introduction to the characters of representations of Brauer algebra. The
main idea is that although we don't have the notion of conjugacy class, we
can define a notion of character class which plays a similar role as
conjugacy class. We summarize some main results here. A character class is
represented by a staircase $\zeta =(h,\zeta ^{+},\zeta ^{-})$. And for an
element $d$, its character class is determined by its cycle type. The
characters of elements in the same character class have the same value up to
a factor of exponent of $N$. In the equation (\ref{character}) of the
character, we have a sum over $\delta \vdash (l-h)$, which tells us that the
character $\chi ^{\gamma }(\zeta )=0$ for $h>l$. Besides, the character is
of order $N^{z(d)}$, where $z(d)$ by definition is the number of zero cycles
in $d$. And we must have $z(d)\leq h(d)$. Therefore we have an inequality $%
z(d)\leq h\leq l$. So we have$~\frac{1}{N^{l}}\chi ^{\gamma }(d)\geq
O(N^{0}) $ if and only if $z(d)=h=l$. Furthermore, this condition $z(d)=h=l$
would mean that $d$ has exactly $l$ zero cycles, when each zero cycle only
contains one vertex in each side of the wall, and all other cycles that is
not a zero cycle must be completely contained in only one side of the wall.
Take the diagram (\ref{Example_diagram}) as an example, it does not satisfy
this condition since it has a cycle $6,4^{\prime },3^{\prime }$, which is
not a zero cycle and contains vertices from both side of the wall.

From our previous discussion, we know that $\frac{1}{N^{n+m}}\langle \text{tr%
}(\alpha Z^{\otimes n}\otimes Y^{\otimes m})^{\dagger }N^{l}O^{\gamma
}\rangle \geq O(N^{0})$ if and only if $b_{\alpha }=\Sigma ^{-1}(\alpha
^{-1})$ has $l$ zero cycles, with each zero cycle only containing one vertex
in each side of the wall, and each non zero cycle is contained completely in
one side of the wall. On the other hand, $\text{tr}(\alpha Z^{\otimes
n}\otimes Y^{\otimes m})$ is determined by the $S_{n}\times S_{m}$
equivalence class of $\alpha $ which we denoted by $[\alpha ]$. Therefore we
need to know how to start from the class $[\alpha ]$ to obtain the cycle
type of $b_{\alpha }$.

We know that the class $[\alpha ]$ is described by a sequence $\{w_{[\vec{k}%
]}\}_{[\vec{k}]\in \tilde{K}}$. And we will describe a general procedure to
construct the cycle type of $b_{\alpha }$ from the sequence $\{w_{[\vec{k}%
]}\}_{[\vec{k}]\in \tilde{K}}$. We need to translate the above condition on $%
b_{\alpha }$ into a set of conditions on the sequence $\{w_{[\vec{k}]}\}.$

We obtain the following lemma and give its proof.

\begin{lemma}
The condition that $b_{\alpha }$ has $l$ zero cycles, and each zero cycle
only contains one vertex in each side of the wall, while each non zero cycle
is contained completely in one side of the wall, can be equivalently
described by the condition on the sequence $\{w_{[\vec{k}]}\}_{[\vec{k}]\in
\tilde{K}}$ (or on $\{w_{[f]}\}_{[f]\in \tilde{V}}$) associated to the class
$[\alpha ]$. The condition is

\begin{enumerate}
\item[(1)] $w_{(1,1)} = l$

\item[(2)] $w_{[\vec{k}]} = 0$ for $[\vec{k}] \notin \tilde{K}_0\bigcup
\{(1,1)\} $
\end{enumerate}

\label{cycle_condition}

or equivalently

\begin{enumerate}
\item[(1')] $w_{[\uparrow \downarrow]} = l$

\item[(2')] $w_{[f]}=0$ for $[f]\notin \bigcup_{n}\tilde{V}_{n,0}\bigcup_{m}%
\tilde{V}_{0,m}\bigcup \tilde{V}_{1,1}$

where $\uparrow \downarrow $ is defined to be the map in $V_{1,1}$ that
takes $1$ to $\uparrow $ and $2$ to $\downarrow $.
\end{enumerate}
\end{lemma}

\begin{proof}
Just as $d \in \mathcal{D}_{n,m}$ can be drawn as a diagram, elements of $%
S_{n+m}$ can also be drawn as diagrams. For example
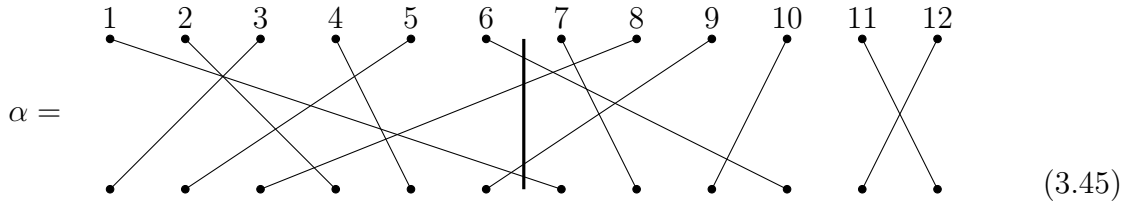
\begin{equation}
\begin{tikzpicture} \node at (0,2) {$\alpha = $}; \node [above] at (1,3)
{$1$}; \node [above] at (2,3) {$2$}; \node [above] at (3,3) {$3$}; \node
[above] at (4,3) {$4$}; \node [above] at (5,3) {$5$}; \node [above] at (6,3)
{$6$}; \node [above] at (7,3) {$7$}; \node [above] at (8,3) {$8$}; \node
[above] at (9,3) {$9$}; \node [above] at (10,3) {$10$}; \node [above] at
(11,3) {$11$}; \node [above] at (12,3) {$12$}; \draw[fill] (1,1) circle
[radius=0.05]; \draw[fill] (2,1) circle [radius=0.05]; \draw[fill] (3,1)
circle [radius=0.05]; \draw[fill] (4,1) circle [radius=0.05]; \draw[fill]
(5,1) circle [radius=0.05]; \draw[fill] (6,1) circle [radius=0.05];
\draw[fill] (1,3) circle [radius=0.05]; \draw[fill] (2,3) circle
[radius=0.05]; \draw[fill] (3,3) circle [radius=0.05]; \draw[fill] (4,3)
circle [radius=0.05]; \draw[fill] (5,3) circle [radius=0.05]; \draw[fill]
(6,3) circle [radius=0.05]; \draw[fill] (7,1) circle [radius=0.05];
\draw[fill] (8,1) circle [radius=0.05]; \draw[fill] (9,1) circle
[radius=0.05]; \draw[fill] (10,1) circle [radius=0.05]; \draw[fill] (11,1)
circle [radius=0.05]; \draw[fill] (12,1) circle [radius=0.05]; \draw[fill]
(7,3) circle [radius=0.05]; \draw[fill] (8,3) circle [radius=0.05];
\draw[fill] (9,3) circle [radius=0.05]; \draw[fill] (10,3) circle
[radius=0.05]; \draw[fill] (11,3) circle [radius=0.05]; \draw[fill] (12,3)
circle [radius=0.05]; \draw[very thick] (6.5,1) to (6.5,3); \draw (2,3) to
(4,1); \draw (3,3) to (1,1); \draw (4,3) to (5,1); \draw (5,3) to (2,1);
\draw (1,3) to (7,1); \draw (6,3) to (10,1); \draw (7,3) to (8,1); \draw
(8,3) to (3,1); \draw (9,3) to (6,1); \draw (10,3) to (9,1); \draw (11,3) to
(12,1); \draw (12,3) to (11,1); \end{tikzpicture}  \label{Example_diagram_01}
\end{equation}

We follow the procedure
\begin{enumerate}
\item[(1)] Start with vertex $t^L_1(\alpha)$ if it exists.

\item[(2)] Follow the edge connected to this vertex. Upon reaching the other
side of the edge, jump to the vertex directly above it if we are in $%
b(\alpha)$ or to the vertex below it if we are in $t(d)$, and continue
following the edge connected to that vertex.

\item[(3)] Following the above procedure, we will end by returning to the
starting vertex and complete a cycle in $\alpha$. We denote such a cycle $%
c_1 $.

\item[(4)] We start from another vertex that has not been visited and repeat
the above process. Each time we finish the above process we will get a cycle
$c_i$ in $\alpha$. And we end the process if we visited all vertices of $%
\alpha $.
\end{enumerate}

In this way, we decompose $\alpha$ into disjoint cycles. For example in the
above diagram (\ref{Example_diagram_01}), we have 4 disjoint cycles. The first is
on vertices $1,7,8,3$, the second on vertices $2,4,5$, the third on $6,10,9$
and the fourth on $11,12$. Note that this decomposition is just the cycle
decomposition of $\alpha$. In the example, the permutation is just $\alpha =
(3,1,7,8)(2,4,5)(6,10,9)(11,12)$.

In each cycle, there are vertex on the left hand side of the wall and vertex
on the right hand side of the wall, depend on $i \leq n$ or $i > n$. We
label a cycle $(i_1,i_2,i_3,\dots)$ by a sequence $(k_1,k_2,k_3,k_4,\dots)$
in such a way that there are first $k_1$ vertex $i_1,i_2,\dots i_{k_1}$ on
the left hand side of the wall and followed by $k_2$ vertex on the right
hand side of the wall $i_{k_1 + 1}, \dots ,i_{k_1+k_2}$ and so on. For
example, cycle $(3,1,7,8)$ is labeled by $(k_1,k_2) = (2,2)$ and $(2,4,5)$
is labeled by $k_1 = 3$. In this way, each cycle is labeled by a $[\vec{k}]
\in \tilde{K}$. Also, we can label a cycle $(i_1,i_2,i_3,\dots)$ by a $[f]
\in \tilde{V}$, where the corresponding $[f]$ is defined in the following
way. We let $f(j) = \uparrow$ if $i_j$ is on the left hand side of the wall
and $f(j) = \downarrow$ if $i_j$ is on the right hand side of the wall. If a
permutation $\alpha$ has the properties that it has $w_{[\vec{k}]}$ cycles
labeled by $[\vec{k}] \in \tilde{K}$ or has $w_{[f]}$ cycles labeled by $[f]$%
, then the operator satisfies
\begin{equation}
\text{tr}(\alpha Z^{\otimes n}\bigotimes Y^{\otimes m}) = \prod_{[\vec{k}%
]\in \tilde{K}}\text{tr}(Z^{k_1}Y^{k_2}Z^{k_3}Y^{k_4} \cdots)^{w_{[\vec{k}%
]}} = \prod_{[f]\in \tilde{V}} O_{[f]}^{w_{[f]}}.
\end{equation}
In the above example
\begin{equation}
\text{tr}(\alpha Z^{\otimes n}\bigotimes Y^{\otimes m}) = \text{tr}(Z^2Y^2)%
\text{tr}(Z^3)\text{tr}(ZY^2)\text{tr}(Y^2).
\end{equation}

Recall that we defined $b_{\alpha} = \Sigma^{-1}(\alpha^{-1})$. We first
describe the $\alpha^{-1}$. In diagram representation, the diagram of $%
\alpha^{-1}$ is obtained from the diagram by interchange $t(\alpha)$ and $%
b(\alpha)$ and keep the edge. In our example (\ref{Example_diagram_01}), the
inverse is
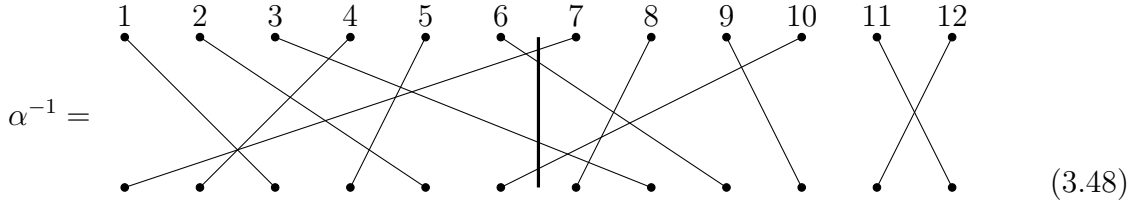
\begin{equation}
\begin{tikzpicture} \node at (0,2) {$\alpha^{-1} = $}; \node [above] at
(1,3) {$1$}; \node [above] at (2,3) {$2$}; \node [above] at (3,3) {$3$};
\node [above] at (4,3) {$4$}; \node [above] at (5,3) {$5$}; \node [above] at
(6,3) {$6$}; \node [above] at (7,3) {$7$}; \node [above] at (8,3) {$8$};
\node [above] at (9,3) {$9$}; \node [above] at (10,3) {$10$}; \node [above]
at (11,3) {$11$}; \node [above] at (12,3) {$12$}; \draw[fill] (1,1) circle
[radius=0.05]; \draw[fill] (2,1) circle [radius=0.05]; \draw[fill] (3,1)
circle [radius=0.05]; \draw[fill] (4,1) circle [radius=0.05]; \draw[fill]
(5,1) circle [radius=0.05]; \draw[fill] (6,1) circle [radius=0.05];
\draw[fill] (1,3) circle [radius=0.05]; \draw[fill] (2,3) circle
[radius=0.05]; \draw[fill] (3,3) circle [radius=0.05]; \draw[fill] (4,3)
circle [radius=0.05]; \draw[fill] (5,3) circle [radius=0.05]; \draw[fill]
(6,3) circle [radius=0.05]; \draw[fill] (7,1) circle [radius=0.05];
\draw[fill] (8,1) circle [radius=0.05]; \draw[fill] (9,1) circle
[radius=0.05]; \draw[fill] (10,1) circle [radius=0.05]; \draw[fill] (11,1)
circle [radius=0.05]; \draw[fill] (12,1) circle [radius=0.05]; \draw[fill]
(7,3) circle [radius=0.05]; \draw[fill] (8,3) circle [radius=0.05];
\draw[fill] (9,3) circle [radius=0.05]; \draw[fill] (10,3) circle
[radius=0.05]; \draw[fill] (11,3) circle [radius=0.05]; \draw[fill] (12,3)
circle [radius=0.05]; \draw[very thick] (6.5,1) to (6.5,3); \draw (2,1) to
(4,3); \draw (3,1) to (1,3); \draw (4,1) to (5,3); \draw (5,1) to (2,3);
\draw (1,1) to (7,3); \draw (6,1) to (10,3); \draw (7,1) to (8,3); \draw
(8,1) to (3,3); \draw (9,1) to (6,3); \draw (10,1) to (9,3); \draw (11,1) to
(12,3); \draw (12,1) to (11,3); \end{tikzpicture}
\end{equation}
After the inverse, a cycle of type $[\vec{k}] = (k_1,k_2,\dots,k_{2r -
1},k_{2r})$ becomes a cycle of type $[\vec{k}^{\prime}]= (k_1,k_{2r},k_{2r -
1},k_{2r - 2},\dots,k_3,k_2)$. We thus define a map $\phi: K \to K$ by
\begin{equation}
\phi(k_1,k_2,\dots,k_{2r - 1},k_{2r}) = (k_1,k_{2r},k_{2r - 1},k_{2r -
2},\dots,k_3,k_2).
\end{equation}
That is, a cycle of $\alpha $ labeled by $[\vec{k}]$ becomes a cycle of $%
\alpha ^{-1}$ labeled by $[\phi (\vec{k})]$. Or using the notation of $[f]$,
we define a map $ \phi ^{\prime }: \bigcup_{n,m \geq 0}V_{n,m} \to \bigcup_{n,m \geq 0}V_{n,m}$
by%
\begin{equation}
\phi ^{\prime }(f)(i)=f(m+n+1-i).
\end{equation}%
Then a cycle of $\alpha $ labeled by $[f]$ becomes a cycle of $\alpha ^{-1}$
labeled by $[\phi^{\prime}(f)]$. Note that $\phi(k_1,0) = (k_1,0)$, $\phi(0,k_2) =
(0,k_2)$, $\phi(k_1,k_2) = (k_1,k_2)$. And $%
\phi^{\prime}(\uparrow^n) = \uparrow^n$, $\phi^{\prime}(\downarrow^m) =
\downarrow^m$, $[\phi^{\prime}(\uparrow\downarrow)] = [\uparrow\downarrow]$.

Now we describe $\Sigma^{-1}$. By definition, $\Sigma^{-1}$ interchange $%
t^R(\alpha)$ and $b^{R}(\alpha)$ and keep the edge. In our example
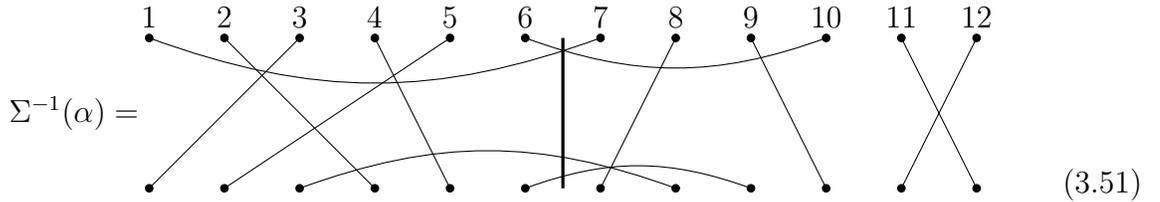
\begin{equation}
\begin{tikzpicture} \node at (0,2) {$\Sigma^{-1}(\alpha) = $}; \node [above]
at (1,3) {$1$}; \node [above] at (2,3) {$2$}; \node [above] at (3,3) {$3$};
\node [above] at (4,3) {$4$}; \node [above] at (5,3) {$5$}; \node [above] at
(6,3) {$6$}; \node [above] at (7,3) {$7$}; \node [above] at (8,3) {$8$};
\node [above] at (9,3) {$9$}; \node [above] at (10,3) {$10$}; \node [above]
at (11,3) {$11$}; \node [above] at (12,3) {$12$}; \draw[fill] (1,1) circle
[radius=0.05]; \draw[fill] (2,1) circle [radius=0.05]; \draw[fill] (3,1)
circle [radius=0.05]; \draw[fill] (4,1) circle [radius=0.05]; \draw[fill]
(5,1) circle [radius=0.05]; \draw[fill] (6,1) circle [radius=0.05];
\draw[fill] (1,3) circle [radius=0.05]; \draw[fill] (2,3) circle
[radius=0.05]; \draw[fill] (3,3) circle [radius=0.05]; \draw[fill] (4,3)
circle [radius=0.05]; \draw[fill] (5,3) circle [radius=0.05]; \draw[fill]
(6,3) circle [radius=0.05]; \draw[fill] (7,1) circle [radius=0.05];
\draw[fill] (8,1) circle [radius=0.05]; \draw[fill] (9,1) circle
[radius=0.05]; \draw[fill] (10,1) circle [radius=0.05]; \draw[fill] (11,1)
circle [radius=0.05]; \draw[fill] (12,1) circle [radius=0.05]; \draw[fill]
(7,3) circle [radius=0.05]; \draw[fill] (8,3) circle [radius=0.05];
\draw[fill] (9,3) circle [radius=0.05]; \draw[fill] (10,3) circle
[radius=0.05]; \draw[fill] (11,3) circle [radius=0.05]; \draw[fill] (12,3)
circle [radius=0.05]; \draw[very thick] (6.5,1) to (6.5,3); \draw (2,3) to
(4,1); \draw (3,3) to (1,1); \draw (4,3) to (5,1); \draw (5,3) to (2,1);
\draw (8,3) to (7,1); \draw (9,3) to (10,1); \draw (11,3) to (12,1); \draw
(12,3) to (11,1); \draw (1,3) to [out=-20,in=200] (7,3); \draw (6,3) to
[out=-20,in=200] (10,3); \draw (3,1) to [out=20,in=160] (8,1); \draw (6,1)
to [out=20,in=160] (9,1); \end{tikzpicture}
\end{equation}
Since $\Sigma^{-1}$ keeps every edge, each cycle $c$ in $\alpha$ is still a
cycle $c$ in $\Sigma^{-1}(\alpha)$. And for a cycle $c$ labeled by $%
(k_1,k_2,\dots)$, the corresponding $type(c)$ can be calculated as
\begin{equation}
type(c) = \sum_{i \geq 1}k_{2i-1} - \sum_{i \geq 1}k_{2i}.
\end{equation}
For the same reason, the type for a cycle $c$ labeled by $[f] \in \tilde{V}_{n,m}$ is just $n
- m$.

Therefore if a cycle in $\Sigma(\alpha)$ is a zero cycle, then $\sum_{i \geq
1}k_{2i-1} - \sum_{i \geq 1}k_{2i} = 0$ (or $n = m $). And if a cycle $c$ in
$\Sigma(\alpha)$ only contains one vertex in each side of the wall, then this
cycle is labeled by $[\vec{k}] = [(1,1)]$ (or [$\uparrow\downarrow$]).

Therefore the condition that $b_{\alpha }=\Sigma ^{-1}(\alpha ^{-1})$ has $l$
zero cycles, and each zero cycle only contains one vertex in each side of the
wall, while each non zero cycle is contained completely in one side of the wall,
is equivalent to
\begin{equation}
w_{\phi (1,1)}=w_{(1,1)}=l,\quad w_{[\vec{k}]}=0\text{ for }[\vec{k}]\notin
\tilde{K}_{0}\bigcup \{(1,1)\},
\end{equation}
or equivalently as
\begin{equation}
w_{[\phi^{\prime}(\uparrow\downarrow)]} = w_{[\uparrow\downarrow]} = l,\quad
w_{[f]} = 0 \text{ for } [f] \notin \bigcup_n\tilde{V}_{n,0}\bigcup_m \tilde{%
V}_{0,m}\bigcup \tilde{V}_{1,1}.
\end{equation}

This is just the condition in the statement of the lemma.

\end{proof}

Therefore we know that a trace product basis state $\prod_{[\vec{k}]\in
K}a_{[\vec{k}]}^{\dagger w_{[\vec{k}]}}\rvert 0\rangle $ has nonzero overlap
with Brauer state $\rvert (l,\gamma ^{+},\gamma ^{-})\rangle $ if and only
if the trace product basis state is of the form
\begin{equation}
\prod_{k\geq 1}a_{k}^{\dagger w_{k}}\prod_{\bar{k}\geq 1}a_{\bar{k}%
}^{\dagger w_{\bar{k}}}a_{(1,1)}^{\dagger l}\rvert 0\rangle ,
\end{equation}%
with the constraints%
\begin{equation}
\quad \sum_{k}kw_{k}=n-l,\quad \sum_{\bar{k}}\bar{k}w_{\bar{k}}=m-l.
\end{equation}

The above two constraints defined two partitions of $n-l,m-l$ respectively,
and we let $[\alpha _{1}],[\alpha _{2}]$ to represent the two partitions
respectively, then we have
\begin{equation}
\langle 0\lvert \prod_{k\geq 1}a_{k}^{w_{k}}\prod_{\bar{k}\geq 1}a_{\bar{k}%
}^{w_{\bar{k}}}a_{(1,1)}^{l}\rvert (l,\gamma ^{+},\gamma ^{-})\rangle =n!m!%
\tilde{t}_{\gamma }\chi ^{\gamma ^{+}}([\alpha _{1}])\chi ^{\gamma
^{-}}([\alpha _{2}])(1+O(1/N)).
\end{equation}%
Using this results, we can write the Brauer state as

\begin{equation}
\rvert (l,\gamma ^{+},\gamma ^{-})\rangle =\frac{n!m!d_{\gamma
^{+}}d_{\gamma ^{-}}}{l!(n-l)!(m-l)!}\rvert \gamma ^{+}\rangle \otimes
\rvert \gamma ^{-}\rangle \otimes a_{(1,1)}^{\dagger l}\rvert 0\rangle
+O(1/N),  \label{Brauer expression}
\end{equation}%
where $\rvert \gamma ^{+}\rangle $ is the Young tableau state associated
with $a_{k}^{\dagger }$, and $\rvert \gamma ^{-}\rangle $ is the Young
tableau state associated with $a_{\bar{k}}^{\dagger }$. Our analysis in the
large $N$ limit shows that to leading order, the Brauer operators can be
expressed as the product of Young tableau operators with an extra operator
mixing two matrices.

We can equivalently write%
\begin{equation}
\langle \lbrack \alpha ]\lvert (l,\gamma ^{+},\gamma ^{-})\rangle =%
\begin{cases}
O(1/N),~~~~\alpha \neq \alpha _{1}\otimes \alpha _{2} \\
\frac{n!m!d_{\gamma ^{+}}d_{\gamma ^{-}}}{l!(n-l)!(m-l)!}\chi ^{\gamma
^{+}}(\alpha _{1})\chi ^{\gamma ^{-}}(\alpha _{2})(1+O(1/N)),~~~~~\alpha
=\alpha _{1}\otimes \alpha _{2}%
\end{cases}%
.
\end{equation}

From equation (\ref{Brauer expression}), we see that the Brauer states $%
\rvert (l,\gamma ^{+},\gamma ^{-})\rangle $ lie in the kernel of dilatation
operator (\ref{D_2_01}) in the infinite $N$ limit. In this case $%
a_{(n,0)}^{\dagger },a_{(0,m)}^{\dagger }$ and $a_{(1,1)}^{\dagger }$ all
commute with the dilatation operator. For non-planar correction of the
action of dilatation operator on more general Brauer operators $%
O_{A;i,j}^{\gamma }$, see \cite{Kimura:2010tx,Kimura:2013fqa}.

\vspace{1pt}

\subsection{Coherent states and their overlaps with Brauer states}

In the rest of this Section, we work in the infinite $N$ limit. We now
consider the overlap between coherent states and Brauer states. We first
consider general coherent states and then quarter BPS coherent states.

First, the Brauer state $\rvert \gamma \rangle $ takes the form
\begin{equation}
\rvert (l,\gamma ^{+},\gamma ^{-})\rangle =\frac{n!m!d_{\gamma
^{+}}d_{\gamma ^{-}}}{l!(n-l)!(m-l)!}\rvert \gamma ^{+}\rangle \otimes
\rvert \gamma ^{-}\rangle \otimes a_{(1,1)}^{\dagger l}\rvert 0\rangle .
\label{Brauer_state}
\end{equation}%
On the other hand, we consider a general coherent state which takes the form
\begin{equation}
\rvert Coh\rangle =\exp (\sum_{[\vec{k}]\in \tilde{K}}c_{[\vec{k}]}a_{[\vec{k%
}]}^{\dagger })\rvert 0\rangle .
\end{equation}%
As in Section 2, we denote
\begin{equation}
c_{k}=c_{(k,0)},\quad c_{\bar{k}}=c_{(0,k)}.
\end{equation}%
Then the coherent state can be factorized into different modes and we mainly
separate the following three parts
\begin{equation}
\exp (\sum_{k}c_{k}a_{k}^{\dagger })=(\sum_{\vec{w}}\prod_{k}\frac{%
(c_{k}a_{k}^{\dagger })^{w_{k}}}{w_{k}!}),
\end{equation}%
\begin{equation}
\exp (\sum_{\bar{k}}c_{\bar{k}}a_{\bar{k}}^{\dagger })=(\sum_{\vec{w}}\prod_{%
\bar{k}}\frac{(c_{\bar{k}}a_{\bar{k}}^{\dagger })^{w_{\bar{k}}}}{w_{\bar{k}}!%
}),
\end{equation}%
and $\exp (c_{(1,1)}a_{(1,1)}^{\dagger })=\sum_{l}\frac{%
(c_{(1,1)}a_{(1,1)}^{\dagger })^{l}}{l!}$.

In the following, we take a specific class of coherent states by letting
\begin{equation}
c_{k}=\frac{\Lambda _{k}}{k}=\frac{x_{1}^{k}+x_{2}^{k}+\dots }{k},\quad c_{%
\bar{k}}=\frac{\Lambda _{\bar{k}}}{\bar{k}}=\frac{y_{1}^{\bar{k}}+y_{2}^{%
\bar{k}}+\dots }{\bar{k}}.
\end{equation}%
Then the overlap between the coherent state and Brauer state can be
calculated as follows:
\begin{eqnarray}
\langle (l,\gamma ^{+},\gamma ^{-})\lvert Coh\rangle  &=&\frac{n!m!d_{\gamma
^{+}}d_{\gamma ^{-}}}{l!(n-l)!(m-l)!}(\langle \gamma ^{+}\lvert \otimes
\langle \gamma ^{-}\lvert \otimes \langle 0\lvert a_{(1,1)}^{l})\rvert
Coh\rangle   \notag \\
&=&\frac{n!m!d_{\gamma ^{+}}d_{\gamma ^{-}}}{l!(n-l)!(m-l)!}\langle \gamma
^{+}\lvert \sum_{\vec{w}}\prod_{k}\frac{1}{w_{k}!}(\frac{\Lambda
_{k}a_{k}^{\dagger }}{k})^{w_{k}}\rvert 0\rangle \langle \gamma ^{-}\lvert
\sum_{\vec{u}}\prod_{\bar{k}}\frac{1}{u_{\bar{k}}!}(\frac{\Lambda _{\bar{k}%
}a_{\bar{k}}^{\dagger }}{\bar{k}})^{u_{\bar{k}}}\rvert 0\rangle   \notag \\
&~&\quad \quad \times \langle 0\lvert a_{(1,1)}^{l}\sum_{j}\frac{%
(c_{(1,1)}a_{(1,1)}^{\dagger })^{j}}{j!}\rvert 0\rangle .
\end{eqnarray}%
In the above derivation, we only need to consider the three parts since we
have
\begin{equation}
\langle (l,\gamma ^{+},\gamma ^{-})\lvert a_{[\vec{k}]}^{\dagger l_{[\vec{k}%
]}}\rvert 0\rangle =0,\text{ for }[\vec{k}]\notin K_{0}\bigcup K_{2}.
\end{equation}%
This means that the value of the overlap between the coherent state and an $%
l=0$ Brauer state only depends on the value of $c_{[\vec{k}]}$ for $[\vec{k}%
]\in K_{0}$.

To further simplify the above formula, we can use results from our previous
paper \cite{Lin:2017dnz} that
\begin{equation}
\langle \gamma ^{+}\lvert \sum_{\vec{w}}\prod_{k}\frac{1}{w_{k}!}(\frac{%
\Lambda _{k}}{k})^{w_{k}}\rvert t_{k}^{w_{k}}\rangle =s_{\gamma
^{+}}(x_{1},x_{2},\dots ),
\end{equation}%
where $\rvert t_{k}^{w_{k}}\rangle =a_{k}^{\dagger w_{k}}\rvert 0\rangle ~$%
and $s_{\gamma ^{+}}$ is the Schur function corresponding to the Young
tableau $\gamma ^{+}$. For more about Schur functions, see \cite{Sagan}.
Similarly, we denote $\rvert s_{\bar{k}}^{u_{\bar{k}}}\rangle =a_{\bar{k}%
}^{\dagger u_{\bar{k}}}\rvert 0\rangle $ and we have
\begin{equation}
\langle \gamma ^{-}\lvert \sum_{\vec{u}}\prod_{\bar{k}}\frac{1}{u_{\bar{k}}!}%
(\frac{\Lambda _{\bar{k}}}{\bar{k}})^{u_{\bar{k}}}\rvert s_{\bar{k}}^{u_{%
\bar{k}}}\rangle =s_{\gamma ^{-}}(y_{1},y_{2},\dots ).
\end{equation}%
Therefore we have
\begin{equation}
\langle (l,\gamma ^{+},\gamma ^{-})\lvert Coh\rangle =\frac{n!m!d_{\gamma
^{+}}d_{\gamma ^{-}}}{l!(n-l)!(m-l)!}c_{(1,1)}^{l}s_{\gamma
^{+}}(x_{1},\dots )s_{\gamma ^{-}}(y_{1},\dots ).
\end{equation}%
The above results may be regarded as a generalization of our previous
results for the overlap between Young tableau states and half BPS coherent
states.

\vspace{1pt}\vspace{1pt}As a special case, we consider $\gamma =(0,\gamma
^{+},\emptyset )$. In this case, we are considering the Brauer algebra $%
B_{N}(n,0)=\mathbb{C}[S_{n}]$. The character of Brauer algebra just becomes
the character of symmetric group. The Brauer state then becomes:
\begin{equation}
\rvert (0,\gamma ^{+},\emptyset )\rangle =d_{\gamma ^{+}}\rvert \gamma
^{+}\rangle ,
\end{equation}%
which is just the Young tableau state. The overlap between the coherent
state and the Brauer state becomes
\begin{equation}
\langle (0,\gamma ^{+},\emptyset )\lvert Coh\rangle =d_{\gamma
^{+}}s_{\gamma ^{+}}(x_{1},x_{2},\dots ).
\end{equation}%
This is just our previous result for the half BPS case \cite{Lin:2017dnz}.

We consider another special case with $l\neq 0$, $\gamma ^{+}=\emptyset
,\gamma ^{-}=\emptyset .~$ In this special case, we first give a different
derivation of the form of the Brauer state $\rvert (l,\emptyset ,\emptyset
)\rangle $. We see that in this case, $l=n=m$. We have expression for the
projector $P^{\gamma }$(see (28) in \cite{Kimura:2010tx}):
\begin{equation}
P^{\gamma }=\frac{1}{N^{l}}\Omega _{l}^{-1}C_{(l)},
\end{equation}%
where $C_{(l)}$ is defined to be
\begin{equation}
C_{(l)}=\sum_{\sigma \in S_{l}}C_{\sigma (1)\bar{1}}C_{\sigma (2)\bar{2}%
}\cdots C_{\sigma (l)\bar{l}}.
\end{equation}%
And we have
\begin{equation}
\text{tr}(C_{\sigma (1)\bar{1}}C_{\sigma (2)\bar{2}}\cdots C_{\sigma (l)\bar{%
l}}Z^{\otimes n}\otimes Y^{\otimes m})=\text{tr}(ZY)^{l}.
\end{equation}%
From our previous results, $\Omega _{l}^{-1}=1+O(1/N)$, we have that $%
N^{l}O^{\gamma }(Z,Y)=l!$tr$(ZY)^{l}+O(1/N)$.\ Then in the infinite $N$
limit,
\begin{equation}
\rvert (l,\emptyset ,\emptyset )\rangle =l!a_{(1,1)}^{\dagger l}\rvert
0\rangle .  \label{Special_case_2}
\end{equation}%
This coincides with our equation (\ref{Brauer_state}) after taking $\gamma
^{+}=\gamma ^{-}=\emptyset $. And the overlap with coherent state is
\begin{equation}
\langle 0\lvert a_{(1,1)}^{l}\rvert Coh\rangle =l!c_{(1,1)}^{l}.
\end{equation}

Another special case is provided for $l\neq 0$, $\gamma ^{-}=\emptyset .~$In
this case, $m-l=0$, we have
\begin{equation}
\rvert (l,\gamma ^{+},\emptyset )\rangle =\frac{n!d_{\gamma ^{+}}}{(n-l)!}%
\rvert \gamma ^{+}\rangle \otimes a_{(1,1)}^{\dagger l}\rvert 0\rangle .
\end{equation}%
And the overlap with coherent state is
\begin{equation}
\langle (l,\gamma ^{+},\emptyset )\lvert Coh\rangle =\frac{n!d_{\gamma ^{+}}%
}{(n-l)!}c_{(1,1)}^{l}s_{\gamma ^{+}}(x_{1},\dots ).
\end{equation}

Now we move on to consider the quarter BPS coherent states $\rvert Coh^{%
\frac{1}{4}}\rangle $, which can be written as
\begin{equation}
\rvert Coh^{\frac{1}{4}}\rangle =\exp (\sum_{n,m\geq
0}c_{(n,m)}a_{(n,m)}^{\dagger })\rvert 0\rangle .
\end{equation}%
\vspace{1pt} The symmetrized operator $a_{(1,1)}^{\dagger }$ is the same as
the operator $a_{[\vec{k}]}^{\dagger }$ for $[\vec{k}]=(1,1)$. Therefore the
overlap between a Brauer state $\rvert \gamma \rangle $ and a quarter BPS
coherent state still takes the form:
\begin{equation}
\langle (l,\gamma ^{+},\gamma ^{-})\lvert Coh^{\frac{1}{4}}\rangle =\frac{%
n!m!d_{\gamma ^{+}}d_{\gamma ^{-}}}{l!(n-l)!(m-l)!}c_{(1,1)}^{l}s_{\gamma
^{+}}(x_{1},x_{2},\dots )s_{\gamma ^{-}}(y_{1},y_{2},\dots ),
\label{Overlap_Coh_Brauer}
\end{equation}%
where
\begin{equation}
c_{(n,0)}=\frac{\Lambda _{n}}{n}=\frac{x_{1}^{n}+x_{2}^{n}+\dots }{n},\quad
c_{(0,m)}=\frac{\Lambda _{m}^{\prime }}{m}=\frac{y_{1}^{m}+y_{2}^{m}+\dots }{%
m}.
\end{equation}

We have considered the inner products $\langle \gamma \lvert Coh^{\frac{1}{4}%
}\rangle $ and $\langle Coh^{\frac{1}{4}}\lvert Coh^{\frac{1}{4}}\rangle $.
We considered similar quantity in our previous work since their quotient
gives us normalized value of the overlap. Therefore, the normalized version
of the overlap $\langle \gamma \lvert Coh^{\frac{1}{4}}\rangle ~$should be
\begin{equation}
\frac{\langle \gamma \lvert Coh^{\frac{1}{4}}\rangle }{\sqrt{\langle \gamma
\lvert \gamma \rangle \langle Coh^{\frac{1}{4}}\rvert Coh^{\frac{1}{4}%
}\rangle }}.
\end{equation}%
And we find%
\begin{eqnarray}
&&\frac{|\langle \gamma \lvert Coh^{\frac{1}{4}}\rangle |^{2}}{\langle
\gamma \lvert \gamma \rangle \langle Coh^{\frac{1}{4}}\rvert Coh^{\frac{1}{4}%
}\rangle }  \notag \\
&=&\frac{1}{l!}\prod_{i,j}(1-x_{i}x_{j}^{\ast
})\prod_{i,j}(1-y_{i}y_{j}^{\ast })\exp (-|c|^{2})|c|^{2l}\left\vert
s_{\gamma ^{+}}(x_{1},\dots )s_{\gamma ^{-}}(y_{1},\dots )\right\vert ^{2}.
\notag \\
&&
\end{eqnarray}

In the half BPS case, $l=0$ and $\gamma ^{-}=\emptyset $, we have the
normalized overlap
\begin{equation}
\prod_{i,j}(1-x_{i}x_{j}^{\ast })\left\vert s_{\gamma ^{+}}(x_{1},\dots
)\right\vert ^{2},
\end{equation}%
which is the same as our previous results in \cite{Lin:2017dnz}, where we
further analyzed the case with rectangular tableaux $\gamma ^{+}=\Box _{LM}$.

Although Brauer state of the form $\rvert (l,\gamma ^{+},\gamma ^{-})\rangle
$ does not span the whole Hilbert space, we can consider the truncated
subspace spanned by $a_{(n,0)}^{\dagger },a_{(0,m)}^{\dagger
},a_{(1,1)}^{\dagger }$. For this reason, we can also consider a subclass of
coherent states defined by
\begin{equation}
\rvert Coh^{\frac{1}{4}}\rangle =\exp (\sum_{n}\frac{\Lambda _{n}}{n}%
a_{(n,0)}^{\dagger }+\sum_{m}\frac{\Lambda _{m}^{\prime }}{m}%
a_{(0,m)}^{\dagger }+ca_{(1,1)}^{\dagger })\rvert 0\rangle .
\end{equation}%
And we use notation $\rvert Coh^{\frac{1}{4}}(x,y,c)\rangle $ by requiring
that
\begin{equation}
\Lambda _{n}=x_{1}^{n}+x_{2}^{n}+\cdots ,\quad \Lambda _{m}^{\prime
}=y_{1}^{m}+y_{2}^{m}+\cdots .
\end{equation}

The advantage of considering the truncated quarter BPS coherent state can be
seen from the following proposition.

\begin{prop}
The Brauer state $\rvert \gamma \rangle =\rvert (l,\gamma ^{+},\gamma
^{-})\rangle $ and the coherent state $\rvert Coh^{\frac{1}{4}%
}(x,y,c)\rangle $ can be transformed into each other by the following
formulas

\begin{equation}
\rvert Coh^{\frac{1}{4}}(x,y,c)\rangle =\sum_{\gamma }\frac{(n-l)!(m-l)!}{%
n!m!d_{\gamma ^{+}}d_{\gamma ^{-}}}c^{l}s_{\gamma ^{+}}(x)s_{\gamma
^{-}}(y)\rvert (l,\gamma ^{+},\gamma ^{-})\rangle  \label{Brauer_to_Coh}
\end{equation}%
and
\begin{eqnarray}
\rvert (l,\gamma ^{+},\gamma ^{-})\rangle &=&\frac{n!m!d_{\gamma
^{+}}d_{\gamma ^{-}}}{(n-l)!(m-l)!}\frac{1}{M_{+}!M_{-}!}\oint \frac{dc}{%
2\pi ic}\prod_{j=1}^{M_{+}}\frac{dx_{j}}{2\pi ix_{j}}\prod_{j=1}^{M_{-}}%
\frac{dy_{j}}{2\pi iy_{j}}  \notag \\
&&\times c^{-l}\prod_{1\leq i<j\leq M_{+}}\left\vert 1-\frac{x_{i}}{x_{j}}%
\right\vert ^{2}\prod_{1\leq i<j\leq M_{-}}\left\vert 1-\frac{y_{i}}{y_{j}}%
\right\vert ^{2}s_{\gamma ^{+}}(x^{-1})s_{\gamma ^{-}}(y^{-1})\rvert Coh^{%
\frac{1}{4}}(x,y,c)\rangle  \notag \\
&&  \label{Coh_to_Brauer}
\end{eqnarray}%
where $M_{+}$ is the number of rows of $\gamma ^{+}$ and $M_{-}$ is the
number of rows of $\gamma ^{-}$. The $x^{-1}$ here is a short hand for $%
(x_{1}^{-1},x_{2}^{-1},\dots )$. And the integration $\oint $ is alone the
circular paths defined by $|c|=1,\;|x_{j}|=1,\;|y_{j}|=1.$
\end{prop}

\begin{proof}
	
	For the first formula, we need to note that Brauer states of the form $\rvert \gamma\rangle$ form an orthogonal basis of the truncated subspace. This is because Young tableau states $\rvert \gamma^+\rangle$ for all possible Young tableaux $\gamma^+$ form an orthonormal basis of the subspace $\bigcup_n\mathcal{H}_{(n,0)}$. And similarly for $\rvert \gamma^-\rangle$. Also, $\{a_{(1,1)}^{\dagger l }\rvert 0 \rangle\}_{l\geq 0}$ form an orthogonal basis of the subspace $\mathcal{H}_{(1,1)}$. So we have
	\begin{equation}
	\rvert Coh^{\frac{1}{4}}(x,y,c) \rangle = \sum_{\gamma}\frac{\langle \gamma\lvert Coh^{\frac{1}{4}}(x,y,c) \rangle }{\langle \gamma \lvert \gamma \rangle} \rvert \gamma\rangle.
	\end{equation}
Then, using the formula (\ref{Overlap_Coh_Brauer}) for the overlap, we can derive equation (\ref{Brauer_to_Coh}).

The derivation of the second formula can be considered as performing inverse Fourier transform to the first one. We need to use the following results:
\begin{eqnarray}
\oint \frac{dc}{2\pi i c}c^{-l}c^{l'} = \delta_{ll'}, \\
\frac{1}{M!}\oint \prod_{j = 1}^{M} \frac{dx_j}{2\pi i x_j}\prod_{1\leq i < j \leq M}\left|1 - \frac{x_i}{x_j}\right|^2s_{\lambda}(x^{-1})s_{\mu}(x) = \delta_{\lambda\mu}.
\end{eqnarray}
The first formula is obvious by calculating the residue, and the second one is from \cite{Prob}, see also \cite{Lin:2017dnz}. We then multiply the three factors $\frac{dc}{2\pi i c}c^{l'}$,  $\prod_{j = 1}^{M} \frac{dx_j}{2\pi i x_j}\prod_{1\leq i < j \leq M}\left|1 - \frac{x_i}{x_j}\right|^2s_{\gamma^{+\prime}}(x^{-1})$ and $\prod_{j = 1}^{M} \frac{dy_j}{2\pi i y_j}\prod_{1\leq i < j \leq M}\left|1 - \frac{y_i}{y_j}\right|^2s_{\gamma^{-\prime}}(y^{-1})$ to both sides of equation (\ref{Brauer_to_Coh}). After integration of $c,x_i,y_i$ over the contour, we can find equation (\ref{Coh_to_Brauer}).
\end{proof}

\vspace{1pt}

\vspace{1pt}

\vspace{1pt}

\vspace{1pt}

\subsection{Entanglement entropy of Brauer states}

From our above results, we see that the Brauer states $\rvert \gamma \rangle
~$span a subspace of the Hilbert space:
\begin{equation}
\mathcal{H}_{\text{Brauer}\{\gamma \}}=\left( \bigotimes_{[k]\in K_{0}}%
\mathcal{H}_{[\vec{k}]}\right) \otimes \mathcal{H}_{(1,1)}=\left(
\bigotimes_{k\geq 1}\mathcal{H}_{(k,0)}\otimes \mathcal{H}_{(0,k)}\right)
\otimes \mathcal{H}_{(1,1)}.
\end{equation}%
In the following, we always assume that $\gamma =(l,\gamma ^{+},\gamma ^{-})$%
. And to simplify the notation, we write $\mathcal{H}_{k}=\mathcal{H}%
_{(k,0)},\;\mathcal{H}_{\bar{k}}=\mathcal{H}_{(0,k)}$. The subspaces $%
\mathcal{H}_{k},\mathcal{H}_{\bar{k}}$ are generated by $t_{k}$ and $s_{k}$
respectively, where we identify
\begin{equation}
t_{k}\leftrightarrow \text{tr}\left( \frac{Z}{\sqrt{N}}\right) ^{k},\quad
s_{k}\leftrightarrow \text{tr}\left( \frac{Y}{\sqrt{N}}\right) ^{k}.
\end{equation}

Similar to the previous work of \cite{Berenstein:2017abm,Lin:2017dnz}, we
define traces
\begin{equation}
\text{tr}_{j}=\text{tr}_{\otimes _{k\neq j}\mathcal{H}_{k}\bigotimes \otimes
_{\bar{k}}\mathcal{H}_{\bar{k}}\otimes \mathcal{H}_{(1,1)}},\quad \text{tr}_{%
\bar{j}}=\text{tr}_{\otimes _{k}\mathcal{H}_{k}\bigotimes \otimes _{\bar{k}%
\neq \bar{j}}\mathcal{H}_{\bar{k}}\otimes \mathcal{H}_{(1,1)}},\quad \text{tr%
}_{(1,1)}=\text{tr}_{\otimes _{k}\mathcal{H}_{k}\bigotimes \otimes _{\bar{k}}%
\mathcal{H}_{\bar{k}}}.
\end{equation}%
The above notation should be distinguished from the trace $\text{tr}_{%
\mathcal{H}_{j}},\text{tr}_{\mathcal{H}_{\bar{j}}}$. We then consider the
entanglement spectrum and entanglement entropy of the Brauer state. After
normalization, the Brauer state is $\frac{1}{\sqrt{l!}}\rvert \gamma
^{+}\rangle \otimes \rvert \gamma ^{-}\rangle \otimes a_{(1,1)}^{\dagger
l}\rvert 0\rangle $. And we write $\rvert l\rangle _{(1,1)}=\frac{1}{\sqrt{l!%
}}a_{(1,1)}^{\dagger l}\rvert 0\rangle $. Then we can calculate the density
operator of a mode of a Brauer state as
\begin{eqnarray}
\hat{\rho}_{j}(\gamma ) &=&\text{tr}_{j}(\rvert \gamma ^{+}\rangle \otimes
\rvert \gamma ^{-}\rangle \otimes \rvert l\rangle _{(1,1)}\langle \gamma
^{+}\lvert \otimes \langle \gamma ^{-}\lvert \otimes \langle l\lvert
_{(1,1)}),  \notag \\
\hat{\rho}_{\bar{j}}(\gamma ) &=&\text{tr}_{\bar{j}}(\rvert \gamma
^{+}\rangle \otimes \rvert \gamma ^{-}\rangle \otimes \rvert l\rangle
_{(1,1)}\langle \gamma ^{+}\lvert \otimes \langle \gamma ^{-}\lvert \otimes
\langle l\lvert _{(1,1)}),  \notag \\
\hat{\rho}_{(1,1)}(\gamma ) &=&\text{tr}_{(1,1)}(\rvert \gamma ^{+}\rangle
\otimes \rvert \gamma ^{-}\rangle \otimes \rvert l\rangle _{(1,1)}\langle
\gamma ^{+}\lvert \otimes \langle \gamma ^{-}\lvert \otimes \langle l\lvert
_{(1,1)}).
\end{eqnarray}%
Then we have
\begin{eqnarray}
\hat{\rho}_{j}(\gamma ) &=&\text{tr}_{j}(\rvert \gamma ^{+}\rangle \otimes
\rvert \gamma ^{-}\rangle \otimes \rvert l\rangle _{(1,1)}\langle \gamma
^{+}\lvert \otimes \langle \gamma ^{-}\lvert \otimes \langle l\lvert
_{(1,1)})  \notag \\
&=&\text{tr}_{j}(\rvert \gamma ^{+}\rangle \langle \gamma ^{+}\lvert
)\langle \gamma ^{-}\lvert \gamma ^{-}\rangle \langle l\rvert l\rangle
_{(1,1)}  \notag \\
&=&\text{tr}_{j}(\rvert \gamma ^{+}\rangle \langle \gamma ^{+}\lvert )
\notag \\
&=&\hat{\rho}_{j}(\gamma ^{+}),
\end{eqnarray}%
where $\hat{\rho}_{j}(\gamma ^{+})$ is the density matrix for the Young
tableau state $\rvert \gamma ^{+}\rangle $. Similarly
\begin{equation}
\hat{\rho}_{\bar{j}}(\gamma )=\hat{\rho}_{\bar{j}}(\gamma ^{-}).
\end{equation}%
And for the $(1,1)$ mode, we have $\hat{\rho}_{(1,1)}(\gamma )=\rvert
l\rangle _{(1,1)}\langle l\lvert _{(1,1)}$, which is a pure state density
matrix.

Then the calculation of entanglement entropy \cite{Horodecki:2009zz} is
straightforward:
\begin{equation}
s_{j}(\gamma )=-\text{tr}_{\mathcal{H}_{j}}(\hat{\rho}_{j}\log (\hat{\rho}%
_{j})),\;s_{\bar{j}}(\gamma )=-\text{tr}_{\mathcal{H}_{\bar{j}}}(\hat{\rho}_{%
\bar{j}}\log (\hat{\rho}_{\bar{j}})),\;s_{(1,1)}(\gamma )=-\text{tr}_{%
\mathcal{H}_{((1,1))}}(\hat{\rho}_{(1,1)}\log (\hat{\rho}_{(1,1)})).
\end{equation}%
We have
\begin{equation}
s_{j}(\gamma )=s_{j}(\gamma ^{+}),\quad s_{\bar{j}}(\gamma )=s_{\bar{j}%
}(\gamma ^{-}),\quad s_{(1,1)}(\gamma )=0,
\end{equation}%
where $s_{j}(\gamma ^{+})$ is the entanglement entropy of the Young tableau
state $\gamma ^{+}$ for mode $j$ and $s_{\bar{j}}(\gamma ^{-})$ is the
entanglement entropy of the Young tableau state $\gamma ^{-}$ for mode $\bar{%
j}$. Detailed analysis of the entanglement entropy $s_{j}(\gamma ^{+})$ has
been given in \cite{Berenstein:2017abm,Lin:2017dnz}. And since $\hat{\rho}%
_{(1,1)}(\gamma )$ is a pure state density matrix, the entanglement entropy
for the mode $(1,1)$ is zero.

\section{Squeezed states and their relation to Brauer states}

\renewcommand{\theequation}{4.\arabic{equation}} \setcounter{equation}{0} %
\renewcommand{\thethm}{4.\arabic{thm}} \setcounter{thm}{0}

Motivated by our previous work \cite{Lin:2017dnz}, we define the squeezed
state as follows
\begin{equation}
\rvert Squ_{n,m;n^{\prime },m^{\prime }}\rangle =\exp (\mu
(a_{(n,m)}^{\dagger }a_{(n^{\prime }m^{\prime })}^{\dagger
}-a_{(n,m)}a_{(n^{\prime }m^{\prime })}))\rvert 0\rangle .
\end{equation}

In the special case where $m=0$ and $m^{\prime }=0$, we can write $%
a_{(n,0)}=a_{n},a_{(n,0)}^{\dagger }=a_{n}^{\dagger }$. And the above
definition gives the half BPS squeezed state (6.1) defined in our previous
work \cite{Lin:2017dnz},
\begin{equation}
\rvert Squ_{nn^{\prime }}\rangle =\exp (\mu (a_{n}^{\dagger }a_{n^{\prime
}}^{\dagger }-a_{n}a_{n^{\prime }}))\rvert 0\rangle .
\label{half BPS squeezed}
\end{equation}%
Since the case for $m=m^{\prime }=0$ has been discussed in our previous
work, here we only consider the case $m,m^{\prime },n,n^{\prime }\geq 1$. In
this case, we use the commutation relation
\begin{equation}
\lbrack a_{(n,m)},a_{(n^{\prime },m^{\prime })}^{\dagger }]=\delta
_{nn^{\prime }}\delta _{mm^{\prime }}\sum_{[f]\in \tilde{V}_{n,m}}\left(
\frac{n!m!}{(n+m)!}\frac{(n+m)}{|Stab_{f}(\mathbb{Z}_{n+m})|}\right) ^{2}.
\end{equation}%
To simplify the calculation, we write the above formula as
\begin{equation}
\lbrack \frac{1}{\kappa _{(n,m)}}a_{(n,m)},\frac{1}{\kappa _{(n^{\prime
},m^{\prime })}}a_{(n^{\prime },m^{\prime })}^{\dagger }]=\delta
_{nn^{\prime }}\delta _{mm^{\prime }},
\end{equation}%
where
\begin{equation}
\kappa _{(n,m)}=\sqrt{\sum_{[f]\in \tilde{V}_{n,m}}\left( \frac{n!m!}{(n+m)!}%
\frac{(n+m)}{|Stab_{f}(\mathbb{Z}_{n+m})|}\right) ^{2}}.
\end{equation}%
Then for $(n,m)\neq (n^{\prime },m^{\prime })$, the squeezed state can be
expanded as follows:
\begin{equation}
\rvert Squ\rangle =\left( 1-\tanh ^{2}(\mu \kappa _{(n,m)}\kappa
_{(n^{\prime },m^{\prime })})\right) ^{\frac{1}{2}}\sum_{j=0}^{\infty }\frac{%
1}{j!}\left( \frac{\tanh (\mu \kappa _{(n,m)}\kappa _{(n^{\prime }m^{\prime
})})}{\kappa _{(n,m)}\kappa _{(n^{\prime },m^{\prime })}}\right)
^{j}a_{(n,m)}^{\dagger j}a_{(n^{\prime },m^{\prime })}^{\dagger j}\rvert
0\rangle .  \label{squeezed state expansion}
\end{equation}

As a remark, we mention that the above formula also applies to the case when
some of $n,m,n^{\prime },m^{\prime }$ equal to zero. We only need to define
\begin{equation}
\kappa _{(n,0)}=\sqrt{n},\quad \kappa _{(0,m)}=\sqrt{m}.
\end{equation}%
And for the case $m=m^{\prime }=0$, the above results give the expansion of
the half BPS squeezed state (\ref{half BPS squeezed}). It's also easy to see
from the above expansion that squeezed states have inner products given by
\begin{equation}
\langle Squ_{n_{1},m_{1};n_{2},m_{2}}\rvert
Squ_{n_{3},m_{3};n_{4},m_{4}}\rangle =\delta _{n_{1}n_{3}}\delta
_{n_{2}n_{4}}\delta _{m_{1}m_{3}}\delta _{m_{2}m_{4}}+\delta
_{n_{1}n_{4}}\delta _{n_{2}n_{3}}\delta _{m_{1}m_{4}}\delta _{m_{2}m_{3}}.
\end{equation}

One of the motivations for us to consider the squeezed state is that we can
take a limit of the squeezed states to obtain EPR states. First define a new
parameter
\begin{equation}
q=\tanh (\mu \kappa _{(n,m)}\kappa _{(n^{\prime },m^{\prime })}).
\end{equation}%
Then the squeezed state can be written as:
\begin{equation}
\rvert Squ_{n,m;n^{\prime },m^{\prime }}\rangle =\left( 1-q^{2}\right) ^{%
\frac{1}{2}}\sum_{j=0}^{\infty }\frac{1}{j!}\left( \frac{q}{\kappa
_{(n,m)}\kappa _{(n^{\prime },m^{\prime })}}\right) ^{j}a_{(n,m)}^{\dagger
j}a_{(n^{\prime },m^{\prime })}^{\dagger j}\rvert 0\rangle .
\end{equation}%
And we can write the corresponding EPR limit as follows:
\begin{equation}
\rvert EPR_{n,m;n^{\prime },m^{\prime }}\rangle =\lim_{q\rightarrow 1}\rvert
Squ_{n,m;n^{\prime },m^{\prime }}\rangle =\mathcal{N}^{-\frac{1}{2}%
}\sum_{j=0}^{\infty }\frac{1}{j!}\left( \frac{1}{\kappa _{(n,m)}\kappa
_{(n^{\prime },m^{\prime })}}\right) ^{j}a_{(n,m)}^{\dagger j}a_{(n^{\prime
},m^{\prime })}^{\dagger j}\rvert 0\rangle ,
\end{equation}%
where $\mathcal{N}^{-\frac{1}{2}}$ is a normalization factor that tends to
infinity as $q\rightarrow 1$. One can take an infinitesimal positive cutoff $%
\epsilon \rightarrow 0$, such that $1-q=\epsilon $ and $\mathcal{N}=\frac{1}{%
2\epsilon }$.

It is interesting to consider the overlap between a squeezed state and a
Brauer state. We can see from the expansion (\ref{squeezed state expansion})
that the overlap is zero when $n,m,n^{\prime },m^{\prime }>1$. And it is not
zero only for the following situations:

\begin{enumerate}
\item[1.] $(n,m)=(1,1),$ $m^{\prime }=0$ and $\gamma =(l,\gamma
^{+},\emptyset )$ with $\gamma ^{+}\vdash ln^{\prime }.$

In this situation,
\begin{eqnarray}
\langle (l,\gamma ^{+},\emptyset )\rvert Squ_{1,1;n^{\prime },0}\rangle
=\left( 1-\tanh ^{2}(\mu \sqrt{n^{\prime }})\right) ^{\frac{1}{2}}\left(
\frac{\tanh (\mu \sqrt{n^{\prime }})}{\sqrt{n^{\prime }}}\right) ^{l}\frac{%
(n^{\prime }l+l)!}{(n^{\prime }l)!}d_{\gamma ^{+}}\chi _{\gamma ^{+}}(\vec{w}%
)|_{\substack{ w_{k}=0\text{ except}  \\ w_{n^{\prime }}=l}}  \notag \\
\end{eqnarray}%
Similarly we can consider $(n,m)=(1,1)$, $n^{\prime }=0$ and $\gamma
=(l,\emptyset ,\gamma ^{-})$ with $\gamma ^{-}\vdash lm^{\prime }$.

\item[2.] $m=0$, $m^{\prime }=0$, and $\gamma =(0,\gamma ^{+},\emptyset )$
with $\gamma ^{+}\vdash j(n+n^{\prime }).$

In this situation, we have that,%
\begin{eqnarray}
\langle (0,\gamma ^{+},\emptyset )\rvert Squ_{n,0;n^{\prime },0}\rangle
=\left( 1-\tanh ^{2}(\mu \sqrt{nn^{\prime }})\right) ^{\frac{1}{2}}\frac{1}{%
j!}\left( \frac{\tanh (\mu \sqrt{nn^{\prime }})}{\sqrt{nn^{\prime }}}\right)
^{j}d_{\gamma ^{+}}\chi _{\gamma ^{+}}(\vec{w})|_{\substack{ w_{k}=0\text{
except}  \\ w_{n}=w_{n^{\prime }}=j}}  \notag \\
\end{eqnarray}%
This situation is the same as (6.7) in \cite{Lin:2017dnz} except that the
normalization of the Brauer state is different.

Similarly we can consider $n=0$, $n^{\prime }=0$, and $\gamma =(0,\emptyset
,\gamma ^{-})$.

\item[3.] $m=0$, $n^{\prime }=0$, and $\gamma =(0,\gamma ^{+},\gamma ^{-})~$%
with $\gamma ^{+}\vdash jn,\gamma ^{-}\vdash jm^{\prime }.$

We have that,
\begin{eqnarray}
&&\langle (0,\gamma ^{+},\emptyset )\rvert Squ_{n,0;0,m^{\prime }}\rangle
\notag \\
&=&\left( 1-\tanh ^{2}(\mu \sqrt{nm^{\prime }})\right) ^{\frac{1}{2}}\frac{1%
}{j!}\left( \frac{\tanh (\mu \sqrt{nm^{\prime }})}{\sqrt{nm^{\prime }}}%
\right) ^{j}d_{\gamma ^{+}}d_{\gamma ^{-}}\chi _{\gamma ^{+}}(\vec{w})|
_{\substack{ w_{k}=0\text{ except}  \\ w_{n}=j}}\chi _{\gamma ^{-}}(\vec{w}%
)| _{\substack{ w_{k}=0\text{ except}  \\ w_{m^{\prime }}=j}}  \notag \\
&&
\end{eqnarray}
\end{enumerate}

\section{Discussion}

\label{sec_Discussion}

\vspace{1pt}

In this paper, we constructed quarter BPS coherent states. The construction
starts with a general construction of the Hilbert space of two-matrix gauge
invariant operators. Then we consider the one-loop dilation operator. In our
case, we care about the kernel of the dilatation operator, and this gives us
the quarter BPS operators. Then the construction of quarter BPS coherent
states generalize the construction of half BPS coherent states by taking
exponential of the creation operators. These quarter BPS coherent states are
also the eigenstates of the annihilation operators. We also computed the
inner products of the quarter BPS coherent states.

\vspace{1pt}

The Brauer operators are also explored in this paper. The construction of
Brauer operators involves characters of irreducible representations of
Brauer algebra \cite{Kimura:2007wy,Kimura:2008ac}. And we calculated the
inner product between Brauer operators and trace product operators. The
construction of Brauer operators and the analysis of them have been carried
out in many previous works, see \cite{Kimura:2007wy,Kimura:2008ac}, and the
explicit form of these operators are known in special cases, see for example
\cite{Kimura:2009ur,Kimura:2010tx,Kimura:2011df}. We see that the Brauer
operators are in some sense the generalization of Young tableau operators in
quarter BPS case. This observation is also closely related to the dual
gravity interpretation \cite{Kimura:2011df}, where the droplet configuration
of the dual gravity solution is described by the two Young tableaux in the
Brauer basis. We also calculated the entanglement entropy of the Brauer
states, and the results are very similar to the Young tableau states.

\vspace{1pt}

One of the motivations of constructing the coherent states is that they are
important ingredients in the study of superposition-induced topology change
in quantum gravity \cite{Berenstein:2017abm,Berenstein:2016pcx}. With our
previous work of a superposition formula that gives a Young tableau state by
superposing half BPS coherent states, we considered here similar
superposition formulas involving Brauer states and quarter BPS coherent
states. We show that one can superpose quarter BPS coherent states to obtain
Brauer states. Conversely, our superposition formulas show that one can also
superpose Brauer states to obtain quarter BPS coherent states. The ideas of
superposition of states on the gravity side have also been considered in
\cite%
{Berenstein:2017abm,Almheiri:2016blp,Nomura:2016aww,Berenstein:2017rrx,Lin:2017dnz}%
. It is useful to explore these ideas with the setup of this paper. Also,
inspired by previous works \cite{Berenstein:2017abm,Lin:2017dnz}, it is very
interesting to further study the relation between entanglement and the dual
spacetime geometry.

\vspace{1pt}

We also generalized the squeezed states from our previous half BPS case \cite%
{Lin:2017dnz} to quarter BPS case. The squeezed states itself can be
regarded as a generalization of the coherent states, since they both satisfy
the property that they can saturate the uncertainty principle. Moreover,
taking certain limit of the squeezed state can give us a EPR pair state,
which is important in the quantum information theory and quantum optics. And
in our setup, it is interesting to study their entanglement properties and
the dual geometric picture.

\vspace{1pt}

In the context of gauge/gravity correspondence, coherent states have gravity
dual descriptions in terms of semiclassical geometries, and this has been
studied in detail in the half BPS case. These coherent states, in the dual
gravity side, correspond to creating deformations \cite%
{Berenstein:2005aa,Grant:2005qc,Mandal:2005wv,Skenderis:2007yb,Takayama:2005yq}
on the vacuum geometry. Some classes of these geometries can be reduced to
lower dimensions and viewed as geometries in lower dimensional gravity \cite%
{Lin:2004nb,Chong:2004ce,Chen:2007du,Liu:2007xj}. Geometries in lower
dimensional gravity that are dual to coherent states have also been
considered in \cite{Gentle:2013fma}. As similar to the half BPS case, there
are smooth spacetime geometries dual to quarter BPS states, see e.g. \cite%
{Chong:2004ce,Chen:2007du,Lunin:2008tf,Gauntlett:2006ns,Gava:2007qs} and
related discussions. These quarter BPS states include the quarter BPS
coherent states that we describe in this paper. It would also be interesting
to explore the gravity dual of the BPS coherent states further.

\vspace{1pt}

Our results may provide further insights into emergent spacetime geometry
and other interesting phenomena in gauge/gravity correspondence. Various
other similar spacetime geometries in the context of string theory and
quantum gravity have been analyzed, see for example \cite{Mathur:2005ai}$-$%
\cite{Fareghbal:2008ar} and their related discussions. Our methods and
discussions may also be related to 2d Yang-Mills \cite{Dijkgraaf:2005bp} and
to fuzzball proposal \cite{Mathur:2005ai}. It would also be good to
understand more the relation to proposals of emergent spacetime geometries.

\vspace{1pt}

We know that the dynamics of half BPS sector of $\mathcal{N}=4$ SYM is
described by a single matrix quantum mechanical model with harmonic
oscillator potential, which itself is equivalent to the dynamics of $N$ free
fermions. And the dynamics of quarter and eighth BPS sector are investigated
in, for example \cite{Berenstein:2005aa}. Therefore, our discussions are
also related to the matrix model approach and other approaches for several
matrix fields \cite%
{Berenstein:2005aa,Masuku:2009qf,Berenstein:2008eg,Honda:2013nfa,deMelloKoch:2016whh,BenGeloun:2017vwn,Harmark:2014mpa}%
.

\vspace{1pt}

We take the approach that first includes both BPS states and non-BPS states.
Although we mainly studied the BPS states, it is also very interesting to
consider other non-BPS states in this system, such as \cite%
{Carlson:2011hy,deMelloKoch:2012ck,Koch:2011hb}. There are restricted Schur
basis and flavor symmetry basis for example \cite{Brown:2007xh}, which have
their own distinct properties and can be transformed into each other.
Therefore we can also study their relation to our setup.

\vspace{1pt}

\section*{Acknowledgments}

We would like to thank D. Berenstein, B. Chen, R. de Mello Koch, Q. T. Li,
J. Maldacena, S. Ramgoolam, J. Shock, J. Simon, N. Su, N. Wallach, J. Wu,
S.-T. Yau, and P. Zhao for discussions and communications. The work was
supported in part by Yau Mathematical Sciences Center and Tsinghua
University.



\appendix

\section{Orthogonality relation of Brauer states}

\label{Appendix A}

\renewcommand{\theequation}{A.\arabic{equation}} \setcounter{equation}{0} %
\renewcommand{\thethm}{A.\arabic{thm}} \setcounter{thm}{0}

\vspace{1pt}For general $l\geq 0$, we take a normalization of Brauer state $%
\rvert \gamma \rangle \leftrightarrow N^{l}O^{\gamma }(\frac{Z}{\sqrt{N}},%
\frac{Y}{\sqrt{N}})$, where $\gamma =(l,\gamma ^{+},\gamma ^{-})$. We can
write down the orthogonality relation. According to formula (7.15) in \cite%
{Kimura:2007wy}%
\begin{equation}
\langle O^{\gamma _{1}}(Y,Z)^{\dagger }O^{\gamma _{2}}(Z,Y)\rangle
=m!n!\delta _{\gamma _{1},\gamma _{2}}d_{\gamma _{1}}\dim \gamma _{1}.
\end{equation}%
We write the above formula in terms of notation $\rvert \gamma _{1}\rangle
=\rvert (l_{1},\gamma _{1}^{+},\gamma _{1}^{-})\rangle ,\rvert \gamma
_{2}\rangle =\rvert (l_{2},\gamma _{2}^{+},\gamma _{2}^{-})\rangle ,$
\begin{eqnarray}
\langle \gamma _{1}\lvert \gamma _{2}\rangle &=&\frac{1}{N^{n+m}}\langle
N^{l_{1}}O^{\gamma _{1}}(Y,Z)^{\dagger }N^{l_{2}}O^{\gamma _{2}}(Z,Y)\rangle
\notag \\
&=&\frac{1}{N^{n+m-l_{1}-l_{2}}}m!n!\delta _{\gamma _{1},\gamma
_{2}}d_{\gamma _{1}}\dim \gamma _{1}  \notag \\
&=&m!n!\delta _{\gamma _{1},\gamma _{2}}d_{\gamma _{1}}\frac{\dim \gamma _{1}%
}{N^{n+m-2l_{1}}},
\end{eqnarray}%
where $\delta _{\gamma _{1},\gamma _{2}}=\delta _{\gamma _{1}^{+},\gamma
_{2}^{+}}\delta _{\gamma _{1}^{-},\gamma _{2}^{-}}\delta _{l_{1},l_{2}}$. We
can use the formula for calculating $d_{\gamma }~$(equation (3.12) in \cite%
{Kimura:2007wy})
\begin{equation}
d_{\gamma }=\frac{m!n!}{l!(m-l)!(n-l)!}d_{\gamma ^{+}}d_{\gamma ^{-}},\quad
\text{ for }\gamma =(l,\gamma ^{+},\gamma ^{-}).
\end{equation}%
And we use the definition
\begin{equation}
\tilde{t}_{\gamma }=\lim_{N\rightarrow \infty }\frac{\dim \gamma }{N^{n+m-2l}%
}.
\end{equation}%
For general $l$, we also have an expression for $\tilde{t}_{\gamma }$
derived in Sec. 3.1, which is
\begin{equation}
\tilde{t}_{\gamma }=\frac{d_{\gamma ^{+}}d_{\gamma ^{-}}}{(n-l)!(m-l)!}.
\end{equation}

\vspace{1pt}Then we have
\begin{equation}
\langle \gamma _{1}\lvert \gamma _{2}\rangle =\delta _{\gamma _{1},\gamma
_{2}}\frac{m!^{2}n!^{2}}{l_{1}!(m-l_{1})!(n-l_{1})!}d_{\gamma
_{1}^{+}}d_{\gamma _{1}^{-}}\tilde{t}_{\gamma _{1}}(1+O(1/N)).
\end{equation}%
As a special case, for $l_{1}=0$, we have $\tilde{t}_{\gamma _{1}}=\frac{%
d_{\gamma _{1}^{+}}d_{\gamma _{1}^{-}}}{n!m!}$. Insert this in the above
formula we have%
\begin{equation}
\langle \gamma _{1}\lvert \gamma _{2}\rangle =\delta _{\gamma _{1},\gamma
_{2}}d_{\gamma _{1}^{+}}^{2}d_{\gamma _{1}^{-}}^{2}(1+O(1/N)).
\end{equation}

More generally, we have more Brauer states $O_{A,ij}^{\gamma }(Z,Y)$ which
defined previously. We can may identify
\begin{equation}
\rvert \gamma ;A,ij\rangle \leftrightarrow N^{l}O_{A,ij}^{\gamma }(\frac{Z}{%
\sqrt{N}},\frac{Y}{\sqrt{N}}).
\end{equation}%
According to (7.12) in \cite{Kimura:2007wy}
\begin{equation}
\langle O_{A_{1},i_{1}j_{1}}^{\gamma _{1}}(Z,Y)^{\dagger
}O_{A_{2},i_{2}j_{2}}^{\gamma _{2}}(Z,Y)\rangle =\delta _{\gamma _{1},\gamma
_{2}}\delta _{A_{1}A_{2}}\delta _{i_{1}i_{2}}\delta
_{j_{1}j_{2}}d_{A_{1}}\dim \gamma _{1}.
\end{equation}%
Therefore we have
\begin{eqnarray}
\langle \gamma _{1};A_{1},i_{1}j_{1}\lvert \gamma
_{2};A_{2},i_{2}j_{2}\rangle &=&\frac{1}{N^{n+m-l_{1}-l_{2}}}\delta _{\gamma
_{1},\gamma _{2}}\delta _{A_{1}A_{2}}\delta _{i_{1}i_{2}}\delta
_{j_{1}j_{2}}d_{A_{1}}\dim \gamma _{1}  \notag \\
&=&\delta _{\gamma _{1},\gamma _{2}}\delta _{A_{1}A_{2}}\delta
_{i_{1}i_{2}}\delta _{j_{1}j_{2}}d_{A_{1}}\frac{\dim \gamma _{1}}{%
N^{n+m-2l_{1}}}.
\end{eqnarray}

\section{Characters of Brauer algebra}

\renewcommand{\theequation}{B.\arabic{equation}} \setcounter{equation}{0} %
\renewcommand{\thethm}{B.\arabic{thm}} \setcounter{thm}{0}

\label{appendix brauer algebra}

In this section, we present some results about the characters of Brauer
algebra that will be useful in the study of Brauer states. We mainly refer
to \cite{Halverson:1996} in this section. Brauer algebras are extensively
explored also in e.g. \cite{Ram:1995}$-$\cite{Stoll:2016}.

Remember that the Brauer algebra $B_{N}(n,m)$ is not a group, therefore some
familiar results in group representation theory may not hold in this case.
For example, there is no notion of conjugacy class, since there is no
inverse for every element in the algebra. For this reason, we cannot say
that the character takes the same value on a conjugacy class for the Brauer
algebra. However, we have an analogous notion of conjugacy class which \cite%
{Halverson:1996} calls character class that shares similar feature of
conjugacy class in group representation theory.

First we introduce some basic results and fix some notations related to the
Brauer algebra. The Brauer algebra $B_{N}(n,m)$ has a basis given by $(n,m)$
diagrams $d$. We write $\mathcal{D}_{n,m}$ for the set of all $(n,m)$
diagrams. A $d\in \mathcal{D}_{n,m}$ is defined to be a diagram with a
vertical wall between the $n$th and $(n+1)$th vertices such that vertical
edges never cross the wall and horizontal edges always begin and end on
opposite side of the wall. For example, the diagram below is a $(6,6)$
diagram

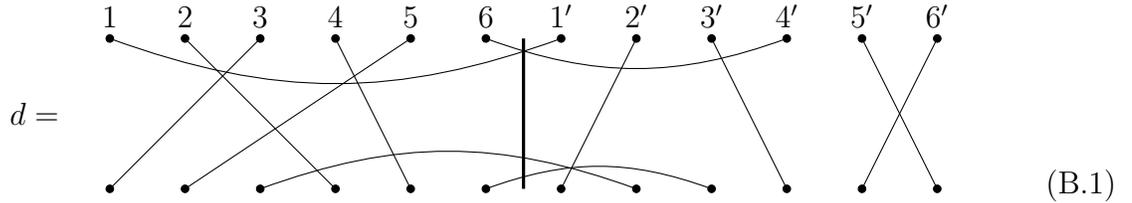
\begin{equation}
\begin{tikzpicture} \node at (0,2) {$d = $}; \node [above] at (1,3) {$1$};
\node [above] at (2,3) {$2$}; \node [above] at (3,3) {$3$}; \node [above] at
(4,3) {$4$}; \node [above] at (5,3) {$5$}; \node [above] at (6,3) {$6$};
\node [above] at (7,3) {$1'$}; \node [above] at (8,3) {$2'$}; \node [above]
at (9,3) {$3'$}; \node [above] at (10,3) {$4'$}; \node [above] at (11,3)
{$5'$}; \node [above] at (12,3) {$6'$}; \draw[fill] (1,1) circle
[radius=0.05]; \draw[fill] (2,1) circle [radius=0.05]; \draw[fill] (3,1)
circle [radius=0.05]; \draw[fill] (4,1) circle [radius=0.05]; \draw[fill]
(5,1) circle [radius=0.05]; \draw[fill] (6,1) circle [radius=0.05];
\draw[fill] (1,3) circle [radius=0.05]; \draw[fill] (2,3) circle
[radius=0.05]; \draw[fill] (3,3) circle [radius=0.05]; \draw[fill] (4,3)
circle [radius=0.05]; \draw[fill] (5,3) circle [radius=0.05]; \draw[fill]
(6,3) circle [radius=0.05]; \draw[fill] (7,1) circle [radius=0.05];
\draw[fill] (8,1) circle [radius=0.05]; \draw[fill] (9,1) circle
[radius=0.05]; \draw[fill] (10,1) circle [radius=0.05]; \draw[fill] (11,1)
circle [radius=0.05]; \draw[fill] (12,1) circle [radius=0.05]; \draw[fill]
(7,3) circle [radius=0.05]; \draw[fill] (8,3) circle [radius=0.05];
\draw[fill] (9,3) circle [radius=0.05]; \draw[fill] (10,3) circle
[radius=0.05]; \draw[fill] (11,3) circle [radius=0.05]; \draw[fill] (12,3)
circle [radius=0.05]; \draw[very thick] (6.5,1) to (6.5,3); \draw (2,3) to
(4,1); \draw (3,3) to (1,1); \draw (4,3) to (5,1); \draw (5,3) to (2,1);
\draw (8,3) to (7,1); \draw (9,3) to (10,1); \draw (11,3) to (12,1); \draw
(12,3) to (11,1); \draw (1,3) to [out=-20,in=200] (7,3); \draw (6,3) to
[out=-20,in=200] (10,3); \draw (3,1) to [out=20,in=160] (8,1); \draw (6,1)
to [out=20,in=160] (9,1); \end{tikzpicture}  \label{Example_diagram}
\end{equation}
We let $t^L_i(d),t^R_j(d)$ denote the $i$th and $j$th vertices in the top on
the right and left side of the wall respectively, as denoted in the above
diagram. And we let $b^L_i(d),b^R_j(d)$ denote the $i$th and $j$th vertices
in the bottom on the right and left side of the wall respectively. We denote
$t(d)$ the set of vertices in the top of the diagram, and $b(d)$ the set of
vertices in the bottom of the diagram.

We then define a cycle type of a diagram $d$ through traversing the diagram $%
d$ as follows:

\begin{enumerate}
\item[(1)] Start with vertex $t^L_1(d)$ if it exists; otherwise start with $%
b_1^R(d)$.

\item[(2)] Follow the edge connected to this vertex. Upon reaching the other
side of the edge, jump to the vertex directly above it if we are in $b(d)$
or to the vertex below it if we are in $t(d)$, and continue following the
edge connected to that vertex.

\item[(3)] Following the above procedure, we will end by returning to the
starting vertex and complete a cycle in $d$. We denote such a cycle $c_1$.

\item[(4)] We start from another vertex that has not been visited and repeat
the above process. Each time we finish the above process we will get a cycle
$c_{i}$ in $d$. And we end the process if we visited all vertices of $d$.
\end{enumerate}

In this way, we decompose $d$ into disjoint cycles. For example in the above
diagram (\ref{Example_diagram}), we have 4 disjoint cycles. The first is on
vertices $1,1^{\prime },2^{\prime },3$, the second on vertices $2,4,5$, the
third on $6,4^{\prime },3^{\prime }$, and the fourth on $5^{\prime
},6^{\prime }$.

For each cycle $c$ in $d$, we define $type(c)$ to be the the number of
vertical edges in $c$ on the left side of the wall minus the the number of
vertical edges in $c$ on the right side of the wall. The integer $type(c)$
is called the cycle type of $c$. We can always reorder all cycles in $d~$in
such a way that
\begin{equation}
type(c_{1})\geq type(c_{2})\geq \cdots \geq type(c_{s}).
\label{cycle_ordered}
\end{equation}%
For example, in the above case (\ref{Example_diagram}), the cycle type of
each cycle in $d$ is
\begin{equation}
type(2,4,5)=3\geq type(1,1^{\prime },2^{\prime },3)=0\geq type(6,4^{\prime
},3^{\prime })=-1\geq type(5^{\prime },6^{\prime })=-2.
\end{equation}%
We then associate with $d\in \mathcal{D}_{n,m}$ a $(n+m)$-staircase $\zeta
(d)$ obtained from (\ref{cycle_ordered}) by inserting $(n+m-s)$ zeros
between the positive values and negative values. That is to say that $\zeta
(d)=(k,\zeta ^{+},\zeta ^{-})$ with $\zeta ^{+}$ the same as the positive
part of $type(c_{i})$ and $\zeta ^{-}$ the same as the negative part of $%
type(c_{i})$. And we call $\zeta (d)$ the cycle type of $d$. For example, in
the above example, $\zeta (d)=(3,0^{9},-1,-2)$. Zero cycles contain the same
number of vertices on each side of the wall. Thus there exists an integer $%
h(d)$ satisfying $\zeta (d)^{+}\vdash (n-h(d))$ and $\zeta (d)^{-}\vdash
(m-h(d))$. In our above example, $n=m=6$, and $h(d)=3$.

The above procedure gives us a way to assign each $d\in \mathcal{D}_{n,m}$ a
$(n,m)$-staircase $\zeta (d)$. We then describe a way to assign each $(n,m)$%
-staircase $\zeta $ a element $d_{\zeta }\in \mathcal{D}_{n,m}$. First for
each $k\in \mathbb{Z}-\{0\}$, we define an element
\begin{equation}
\begin{tikzpicture} \node at (-0.2,1.5) {$d_k = $}; \node [above] at (0.5,2)
{$1$}; \node [above] at (1.5,2) {$2$}; \node [above] at (3,2) {$k-1$}; \node
[above] at (4,2) {$k$}; \draw[fill] (0.5,1) circle [radius=0.05];
\draw[fill] (1.5,1) circle [radius=0.05]; \draw[fill] (2.5,1) circle
[radius=0.05]; \draw[fill] (4,1) circle [radius=0.05]; \draw[fill] (0.5,2)
circle [radius=0.05]; \draw[fill] (1.5,2) circle [radius=0.05]; \draw[fill]
(3,2) circle [radius=0.05]; \draw[fill] (4,2) circle [radius=0.05]; \draw
(0.5,2) to (1.5,1); \draw (1.5,2) to (2.5,1); \draw (3,2) to (4,1); \draw
(4,2) to (0.5,1); \draw [dotted] (2.5,1) to (4,1); \draw [dotted] (1.5,2) to
(3,2); \draw[very thick] (4.5,1) to (4.5,2); \node at (5.5,1.5) {if $k > 0$,
}; \node at (7,1.5) {$d_k = $}; \draw[very thick] (7.5,1) to (7.5,2); \node
[above] at (8,2) {$1'$}; \node [above] at (9,2) {$2'$}; \node [above] at
(10.5,2) {$k^{\prime}-1$}; \node [above] at (11.5,2) {$k'$}; \draw[fill]
(8,1) circle [radius=0.05]; \draw[fill] (9.5,1) circle [radius=0.05];
\draw[fill] (10.5,1) circle [radius=0.05]; \draw[fill] (11.5,1) circle
[radius=0.05]; \draw[fill] (8,2) circle [radius=0.05]; \draw[fill] (9,2)
circle [radius=0.05]; \draw[fill] (10.5,2) circle [radius=0.05]; \draw[fill]
(11.5,2) circle [radius=0.05]; \draw [dotted] (8,1) to (9.5,1); \draw
[dotted] (9,2) to (10.5,2); \node at (12.5,1.5) {if $k^{\prime}=-k>0$.};
\draw (8,1) to (9,2); \draw (9.5,1) to (10.5,2); \draw (10.5,1) to (11.5,2);
\draw (8,2) to (11.5,1); \end{tikzpicture}
\end{equation}%
And we also define element $e$:
\begin{equation}
\begin{tikzpicture} \node at (0,1.5) {$e = $}; \draw[fill] (1,1) circle
[radius=0.05]; \draw[fill] (2,1) circle [radius=0.05]; \draw[fill] (1,2)
circle [radius=0.05]; \draw[fill] (2,2) circle [radius=0.05]; \draw[very
thick] (1.5,1) to (1.5,2); \draw (1,1) to [out=20,in=160] (2,1); \draw (1,2)
to [out=-20,in=200] (2,2); \end{tikzpicture}
\end{equation}%
Now for a $(n,m)$-staircase $\zeta =(\zeta _{1},\zeta _{2},\dots ,\zeta
_{n+m})$, which could also be written as $\zeta =(k,\zeta ^{+},\zeta ^{-})$.
There exist a integer $h(\zeta )$ that $\zeta ^{+}\vdash (n-h(\zeta )),\zeta
^{-}\vdash (m-h(\zeta ))$. And assume that the length of positive and
negative part of $\zeta $ are $l(\zeta ^{+})=i$ and $l(\zeta ^{-})=j$. Then
we define
\begin{equation}
d_{\zeta ^{+}}=d_{\zeta _{1}}\otimes d_{\zeta _{2}}\otimes \cdots \otimes
d_{\zeta _{i}},\quad d_{\zeta ^{+}}=d_{\zeta _{m+n-j}}\otimes \cdots \otimes
d_{\zeta _{m+n-1}}\otimes d_{\zeta _{m+n}}.
\end{equation}%
And we define $d_{\zeta }\in B_{N}(n,m)$ to be
\begin{equation}
d_{\zeta }=d_{\zeta ^{+}}\otimes e^{\otimes h(\zeta )}\otimes d_{\zeta ^{-}}.
\end{equation}%
As an example, let $\zeta =(3,0^{9},-1,-2)$ as a $(6,6)$ staircase. $h(\zeta
)=3$, $\zeta ^{+}=(3)$ and $\zeta ^{-}=(2,1)$. In this case, we have
\begin{equation}
\begin{tikzpicture} \node at (0,2) {$d_{\zeta} = $}; \node [above] at (1,3)
{$1$}; \node [above] at (2,3) {$2$}; \node [above] at (3,3) {$3$}; \node
[above] at (4,3) {$4$}; \node [above] at (5,3) {$5$}; \node [above] at (6,3)
{$6$}; \node [above] at (7,3) {$1'$}; \node [above] at (8,3) {$2'$}; \node
[above] at (9,3) {$3'$}; \node [above] at (10,3) {$4'$}; \node [above] at
(11,3) {$5'$}; \node [above] at (12,3) {$6'$}; \draw[fill] (1,1) circle
[radius=0.05]; \draw[fill] (2,1) circle [radius=0.05]; \draw[fill] (3,1)
circle [radius=0.05]; \draw[fill] (4,1) circle [radius=0.05]; \draw[fill]
(5,1) circle [radius=0.05]; \draw[fill] (6,1) circle [radius=0.05];
\draw[fill] (1,3) circle [radius=0.05]; \draw[fill] (2,3) circle
[radius=0.05]; \draw[fill] (3,3) circle [radius=0.05]; \draw[fill] (4,3)
circle [radius=0.05]; \draw[fill] (5,3) circle [radius=0.05]; \draw[fill]
(6,3) circle [radius=0.05]; \draw[fill] (7,1) circle [radius=0.05];
\draw[fill] (8,1) circle [radius=0.05]; \draw[fill] (9,1) circle
[radius=0.05]; \draw[fill] (10,1) circle [radius=0.05]; \draw[fill] (11,1)
circle [radius=0.05]; \draw[fill] (12,1) circle [radius=0.05]; \draw[fill]
(7,3) circle [radius=0.05]; \draw[fill] (8,3) circle [radius=0.05];
\draw[fill] (9,3) circle [radius=0.05]; \draw[fill] (10,3) circle
[radius=0.05]; \draw[fill] (11,3) circle [radius=0.05]; \draw[fill] (12,3)
circle [radius=0.05]; \draw[very thick] (6.5,1) to (6.5,3); \draw (1,3) to
(2,1); \draw (2,3) to (3,1); \draw (3,3) to (1,1); \draw (10,3) to (11,1);
\draw (11,3) to (10,1); \draw (12,3) to (12,1); \draw (4,1) to
[out=20,in=160] (7,1); \draw (4,3) to [out=-20,in=200] (7,3); \draw (5,1) to
[out=20,in=160] (8,1); \draw (5,3) to [out=-20,in=200] (8,3); \draw (6,1) to
[out=20,in=160] (9,1); \draw (6,3) to [out=-20,in=200] (9,3);
\end{tikzpicture}
\end{equation}

Then the following results tells us that the character of a $d \in \mathcal{D%
}_{n,m}$ of a certain type $\zeta$ is related to the character of the
standard diagram $d_{\zeta}$.

\begin{thm}
{\text{\cite{Halverson:1996}}} Let $d\in \mathcal{D}_{n,m}$ with $\zeta
=\zeta (d)$ and $h=h(d)$. Then for any character $\chi _{B}$ of the Brauer
algebra $B_{N}(n,m)$, we have:
\begin{equation}
\chi _{B}(d)=N^{z(d)-h(d)}\chi _{B}(d_{\zeta }),
\end{equation}%
where $z(d)~\text{is the number of zero-cycles in d.}$
\end{thm}

The above formula tells us that if two $d,d^{\prime }\in \mathcal{D}_{n,m}$
have the same cycle type $\zeta (d)=\zeta (d^{\prime })=\zeta $. Then any
character of Brauer algebra evaluated on the two elements are the same up to
a constant that depends on $N$. For this reason we call the class labeled by
$(n,m)$-staircase $\zeta $ character class.

Now we come to the irreducible representations of Brauer algebra. The
irreducible representation of Brauer algebra is also labeled by $(n,m)$%
-staircase. We have the following formula:
\begin{equation}
\chi ^{\gamma }(d_{\zeta })=N^{h(\zeta )}\sum_{\substack{ \lambda \vdash
n^{\prime }  \\ \pi \vdash m^{\prime }}}\left( \sum_{\delta \vdash
(l-h)}g(\delta ,\gamma ^{+};\lambda )g(\delta ,\gamma ^{-};\pi )\right) \chi
_{S_{n^{\prime }}}^{\lambda }(\zeta ^{+})\chi _{S_{m^{\prime }}}^{\pi
}(\zeta ^{-}),
\end{equation}%
where $\gamma =(l,\gamma ^{+},\gamma ^{-})$, $\zeta =(h,\zeta ^{+},\zeta
^{-})$. And $g$ in the above formula is the Littlewood-Richardson
coefficient. Therefore, for arbitrary $d\in \mathcal{D}_{n,m}$, we have
\begin{equation}
\chi ^{\gamma }(d)=N^{z(d)}\sum_{\substack{ \lambda \vdash n^{\prime }  \\ %
\pi \vdash m^{\prime }}}\left( \sum_{\delta \vdash (l-h)}g(\delta ,\gamma
^{+};\lambda )g(\delta ,\gamma ^{-};\pi )\right) \chi _{S_{n^{\prime
}}}^{\lambda }(\zeta ^{+})\chi _{S_{m^{\prime }}}^{\pi }(\zeta ^{-}),
\label{character}
\end{equation}%
where $\zeta =\zeta (d)$.

\end{document}